\documentclass[11pt, a4paper]{article}
\usepackage[left=25mm,right=25mm,top=25mm,bottom=30mm]{geometry}
\usepackage{framed,latexsym,amssymb,amsmath,amsthm,float,cite,complexity,setspace,color,graphicx,lineno,hyperref,url,lmodern}
\usepackage[shortlabels]{enumitem}
\usepackage[ddmmyyyy]{datetime}
\usepackage[shortend, vlined]{algorithm2e}
\usepackage{tikz,pgfplots}
\pgfplotsset{compat=1.18}

\usepackage{thmtools,thm-restate}
\usepackage{regexpatch}
\makeatletter
\xpatchcmd\thmt@restatable{%
\csname #2\@xa\endcsname\ifx\@nx#1\@nx\else[{#1}]\fi
}{%
\ifthmt@thisistheone
\csname #2\@xa\endcsname\ifx\@nx#1\@nx\else[{#1}]\fi
\else
\csname #2\@xa\endcsname[{restated}]
\fi}{}{}
\makeatother

\newtheorem{theorem}{Theorem}
\newtheorem{lemma}[theorem]{Lemma}
\newtheorem{corollary}[theorem]{Corollary}
\newtheorem{claim}{Claim}

\theoremstyle{definition}
\newtheorem{remark}[theorem]{Remark}
\newtheorem{question}{Question}

\floatstyle{boxed}
\newfloat{problem}{tbhH}{prb}

\tikzset{
	dot/.style = {circle, fill, minimum size=#1,
		inner sep=0pt, outer sep=0pt},
	dot/.default = 5pt
}
\usetikzlibrary{calc}
\usetikzlibrary{tikzmark}

\begin{document}
\onehalfspace
%\linenumbers

\title{Cutwidth and Crossings}
\author{Johannes Rauch$^1$ \and Dieter Rautenbach$^1$}
\date{}
\maketitle
\vspace{-10mm}
\begin{center}
	{\small 
		$^1$ Institute of Optimization and Operations Research, Ulm University, Ulm, Germany\\
		\texttt{\{johannes.rauch, dieter.rautenbach\}@uni-ulm.de}
	}
\end{center}

\begin{abstract}
We provide theoretical insights around the cutwidth of a graph and the One-Sided Crossing Minimization (OSCM) problem.
OSCM was posed in the Parameterized Algorithms and Computational Experiments Challenge 2024, where the cutwidth of the input graph was the parameter in the parameterized track.
We prove an asymptotically sharp upper bound on the size of a graph in terms of its order and cutwidth.
As the number of so-called unsuited pairs is one of the factors that determine the difficulty of an OSCM instance, we provide a sharp upper bound on them in terms of the order $n$ and the cutwidth of the input graph.
If the cutwidth is bounded by a constant, this implies an $\mathcal{O}(2^n)$-time algorithm, while the trivial algorithm has a running time of $\mathcal{O}(2^{n^2})$.
At last, we prove structural properties of the so-called crossing numbers in an OSCM instance.
\end{abstract}

\section{Introduction}\label{sec:intro}
Each year the Parameterized Algorithms and Computational Experiments (PACE)
Challenge~\cite{pacechallenge} poses at least one optimization problem.
The computer science community is challenged to design and implement computer programs that solve these optimization problems.
Usually, there are multiple tracks: 
\begin{itemize}
\item In the exact track, the solver should find an optimal solution in 30 minutes.
\item In the heuristic track, the solver should find a best possible solution in 5 minutes.
\item The parameterized track is similar to the exact track, but additional information, usually a parameter and a corresponding certificate, is given together with the input.
\end{itemize}
Among others, the goals of the PACE Challenge are to 
``bridge the divide between the theory of algorithm design and analysis, and the practice of algorithm engineering'' and to ``inspire new theoretical developments''\cite{pacechallenge}.
In 2024, they posed the following problem for the exact and heuristic track.

\begin{problem}[H]
\noindent
\textsc{One-Sided Crossing Minimization (OSCM)}\\
\textit{Input:} A bipartite graph $G$ with partite sets $A$ and $B$, and a linear ordering $\pi$ of $A$.\\
\textit{Output:} A linear ordering $\sigma$ of $B$ that minimizes the number of \emph{crossings}, that is, $\sigma$ minimizes the number of pairs of edges $a_1b_1$ and $a_2b_2$ of $G$ with $a_1,a_2 \in A$, $b_1,b_2 \in B$, $\pi(a_1) > \pi(a_2)$, and $\sigma(b_1) < \sigma(b_2)$.
\end{problem}

\noindent
For the parameterized track they posed the following problem.
\begin{problem}[H]
\noindent
\textsc{One-Sided Crossing Minimization[Cutwidth] (OSCM[CW])}\\
\textit{Input:} A bipartite graph $G$ with partite sets $A$ and $B$, and a linear ordering $\pi$ of $V(G) = A \dot{\cup} B$.\\
\textit{Parameter:} The cutwidth of $G$ with respect to $\pi$.\\
\textit{Output:} A linear ordering $\sigma$ of $B$ that minimizes the number of \emph{crossings} as in OSCM.
\end{problem}

\noindent
We refer the reader to Section~\ref{sec:pre} for a definition of unknown notions.
We were indeed inspired by the PACE Challenge and present theoretical result around cutwidth alone, cutwidth and crossings together, and crossings alone.

\medskip
Our first result is with regard to the cutwidth of a graph $G$.
The number of vertices of $G$ is its \emph{order}, and the number of edges is its \emph{size}.
Intuitively, similar to other width-parameters, the size of a graph is bounded by a function of the cutwidth.
In Theorem~\ref{thm:size} we specify this intuition and give two bounds on the size of a graph in terms of its cutwidth.
Let $G$ have order $n$ and cutwidth $w$.
The bound of Theorem~\ref{thm:size}~(i) is for ``sparse'' graphs when $w = o(n^2)$.
It is asymptotically sharp.
The bound of Theorem~\ref{thm:size}~(ii) is for ``dense'' graphs when $w = \Theta(n^2)$.

\begin{theorem}[restate=thmubs]\label{thm:size}
Let the graph $G$ have order $n$, size $m$, and cutwidth $w$.
\begin{enumerate}[(i)]
\item $m \leq n\sqrt{2w} + \mathcal{O}(1)$.
\item If $c \in \left(0, \frac{1}{2}\right]$ is such that $w = c(1-c)n^2$ and $cn \geq 1$, then
\[
m \leq \left( \frac{c(4c^2 + 4c\sqrt{c(1-c)} - 12c - 4\sqrt{c(1-c)} + 9)}{6 \sqrt{c(1-c)}} \right)n^2 + \mathcal{O}(n).
\]
\end{enumerate}
\end{theorem}

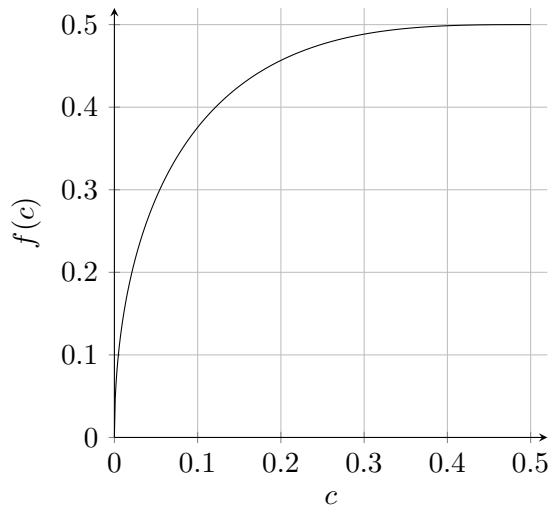
\begin{figure}[h]
\centering
\begin{tikzpicture}
\begin{axis}[
	axis equal image,
	axis lines=left,
	xlabel= \(c\),
	ylabel= \({f(c)}\),
	xmin=0, xmax=0.52,
	ymin=0, ymax=0.52,
	xtick={0, 0.1, 0.2, 0.3, 0.4, 0.5},
	ytick={0, 0.1, 0.2, 0.3, 0.4, 0.5},
	xmajorgrids,
	ymajorgrids
]
\addplot[
	domain=0.000001:0.5,
	samples=500,
	color=black,
	%thick
]
{(x*(4*x^2 + 4*x*sqrt(x*(1-x)) - 12*x - 4*sqrt(x*(1-x)) + 9)/(6*sqrt(x*(1-x)))};
%\addplot[
%	domain=0.000001:0.5,
%	samples=500,
%	color=black,
%	dashed
%]
%{sqrt(x*(1-x))};
\end{axis}
\end{tikzpicture}
\caption{
The function $f(c) = \frac{c(4c^2 + 4c\sqrt{c(1-c)} - 12c - 4\sqrt{c(1-c)} + 9)}{6 \sqrt{c(1-c)}}$ from the bound of Theorem~\ref{thm:size}~(ii) on $\left[0, \frac{1}{2}\right]$.
% compared to $g(c) = \sqrt{c(1-c)}$ (dashed).
}
\end{figure}

\medskip
For our second result we combine cutwidth and crossings and study theoretical aspects of OSCM[CW].
Let $(G = (A \dot{\cup} B, E), \pi)$ be an instance of OSCM or OSCM[CW], and let $u,v \in B$ be two distinct vertices.
The \emph{crossing number} $c_{u,v}$ denotes the number of crossings involving edges incident to $u$ or $v$ in a linear ordering $\sigma$ of $B$ with $\sigma(u) < \sigma(v)$.
Clearly, the total number of crossings only depends on the crossing number for each unordered pair $\{u, v\} \in {B \choose 2}$.
Thus, for a linear ordering $\sigma$ of $B$, the number of crossings is
\[
\sum_{\substack{u,v \in B \\ \sigma(u) < \sigma(v)}} c_{u,v}.
\]
The case $c_{u,v} = c_{v,u} = 0$ occurs if and only if $u$ or $v$ is isolated, or $u$ and $v$ have a single common neighbor.
If $v$ is isolated, then we may solve the instance $(G-v, \pi|_{V(G)-v})$ and insert $v$ anywhere in the obtained linear ordering of $B-v$ for an optimal solution of the instance $(G,\pi)$.
Likewise, if $u$ and $v$ have identical neighborhoods, then we may solve $(G-v, \pi|_{V(G)-v})$ again and insert $v$ right before or after $u$ in the obtained linear ordering of $B-v$ for an optimal solution of $(G,\pi)$.
Therefore, we may assume throughout the article that, in an instance of OSCM or OSCM[CW], there are no isolated vertices in $B$ and no two vertices in $B$ have identical neighborhoods.
This implies $\max \{c_{u,v}, c_{v,u}\} > 0$.

We say that $\{u, v\} \in {B \choose 2}$ is a \emph{suited} (unordered) pair if $c_{u,v} = 0$ or $c_{v,u} = 0$, otherwise $\{u, v\}$ is \emph{unsuited}.
Dujmovi\'c and Whitesides~\cite{dujmovic2004efficient} proved that, 
if $\{u, v\}$ is suited,
then $\{u, v\}$ appears in its \emph{natural ordering} in any optimal solution $\sigma$:
\begin{itemize}
\item If $c_{u,v} = 0$ and $c_{v,u} > 0$, then $\sigma(u) < \sigma(v)$ is the natural ordering.
\item If $c_{u,v} > 0$ and $c_{v,u} = 0$, then $\sigma(u) > \sigma(v)$ is the natural ordering.
%\item If $c_{u,v} = c_{v,u} = 0$, then both $\sigma(u) < \sigma(v)$ and $\sigma(u) > \sigma(v)$ are natural orderings.
\end{itemize}
%The case $c_{u,v} = c_{v,u} = 0$ is irrelevant and can be excluded by preprocessing.
Therefore, if $\{u, v\}$ is a suited pair, we immediately know their position relative to each other in any optimal solution $\sigma$, and an exact algorithm for OSCM or OSCM[CW] only has to determine the relative positions of all unsuited pairs in an optimal linear ordering $\sigma$.
Hence, the number of unsuited pairs is of interest, the trivial bound being $|B| \choose 2$.
In Theorem~\ref{thm:unsuited} we present a tight bound on the number of unsuited pairs in terms of $|B|$ and the cutwidth $w$ of $G$.
\begin{theorem}[restate=thmunsuited]\label{thm:unsuited}
If $(G = (A \dot{\cup} B, E), \pi)$ is an instance of \textup{OSCM[CW]}, and $w$ is the cutwidth of $G$ with respect to $\pi$,
then the number of unsuited pairs is upper bounded by
\[
|B|(w-1) - {w \choose 2},
\]
and this bound is sharp.
\end{theorem}

Dobler~\cite{dobler2023notecomplexityonesidedcrossing} proved that OSCM is \NP-hard on disjoint unions of 4-stars, that is, on graphs of cutwidth at most 2.
Therefore there is no \FPT-algorithm for OSCM[CW] unless $\P = \NP$.
Nevertheless, we can improve the trivial $\mathcal{O}^*(2^{|B| \choose 2})$-time algorithm that enumerates all relative positions of unordered pairs in $B$. 
Corollary~\ref{cor:unsuited} is an immediate consequence of Theorem~\ref{thm:unsuited}.
We use the $\mathcal{O}^*$-notation to suppress polynomial factors in the $\mathcal{O}$-notation.

\begin{corollary}[restate=corunsuited]\label{cor:unsuited}
An instance $(G = (A \dot{\cup} B, E), \pi)$ of \textup{OSCM[CW]}, where $G$ has cutwidth $w$ with respect to $\pi$, is solvable in $\mathcal{O}^*(2^{|B|w})$-time.
\end{corollary}

\medskip
In our third result we shed light on structural properties of crossings.
Let $(G = (A \dot{\cup} B, E), \pi)$ be an instance of OSCM.
Recall that we may assume $\max\{c_{u,v},c_{v,u}\} > 0$ for every two distinct vertices $u, v \in B$.
%If a vertex of $u \in B$ is isolated in $G$, then it is irrelevant to the instance, since $c_{u,v} = c_{v,u} = 0$ for every vertex $v \in B \setminus \{u\}$.
%Therefore we may assume that no vertex in $B$ is isolated.
%This implies $\max\{c_{u,v},c_{v,u}\} > 0$ for every unordered pair $\{u,v\} \in {B \choose 2}$.
Since a linear ordering $\sigma$ of $B$ is uniquely determined by the relative position of each unordered pair of vertices in $B$, we can express $\sigma$ as a transitive tournament.
More explicitly, $\sigma$ can be encoded by a directed graph $D = (B, \overrightarrow{E})$, where $(u,v) \in \overrightarrow{E}$ if and only if $\sigma(u) < \sigma(v)$ for every two distinct vertices $u,v \in B$.
Reversely, a topological ordering of $D$ gives the linear ordering $\sigma$ again.

\begin{algorithm}[h]
\DontPrintSemicolon
\KwIn{An instance $(G = (A \dot{\cup} B, E), \pi)$ of OSCM.}
\KwOut{A linear ordering $\sigma$ of $B$.}
$D \gets (B, \emptyset)$\;
%order all unordered pairs $\{u,v\} \in {B \choose 2}$ non-decreasingly according to $\frac{\min\{c_{u,v}, c_{v,u}\}}{\max\{c_{u,v},c_{v,u}\}}$\;
\ForAll{$\{u,v\} \in {B \choose 2}$ in non-decreasing order with respect to $\frac{\min\{c_{u,v}, c_{v,u}\}}{\max\{c_{u,v},c_{v,u}\}}$}{
	\If{neither $(u,v)$ nor $(v,u)$ is in the transitive closure of $D$}{
		add the arc $(u,v)$ to $D$ if $c_{u,v} \leq c_{v,u}$, otherwise add the arc $(v,u)$ to $D$\;
	}
}
output a topological ordering of $D$\;
\caption{\textsc{Greedy}.}\label{alg:greedy}
\end{algorithm}

Consider Algorithm~\ref{alg:greedy}.
We ask the following.
\begin{question}\label{q:greedy}
Is \textsc{Greedy} a constant-factor approximation algorithm of OSCM?
\end{question}
\noindent 
Closely related to Question~\ref{q:greedy} is Question~\ref{q:cycle}.
\begin{question}\label{q:cycle}
Let $v_1, v_2, \dots, v_t \in B$ be $t \geq 3$ distinct vertices with nonempty neighborhoods.
Assume that the inequality $c_{v_i,v_{i+1}} \leq c_{v_{i+1},v_i}$ holds for every $i \in [t]$, where $v_{t+1}=v_1$.
Note that $c_{v_{i+1},v_i} > 0$ holds for every $i \in [t]$.
What is the largest $c_t \in [0,1]$ such that
\[
\max \left\{ \frac{c_{v_1,v_2}}{c_{v_2,v_1}}, \frac{c_{v_2,v_3}}{c_{v_3,v_2}}, \dots, \frac{c_{v_t,v_1}}{c_{v_1,v_t}} \right\} \geq c_t?
\]
\end{question}

If there is even a constant $c > 0$ such that $c_t \geq c$ for all $t \geq 3$, then \textsc{Greedy} is an $\frac{1}{c}$-factor approximation algorithm.
Thus, the answer to Question~\ref{q:greedy} is ``yes'' in this case.
To see this, let $D = (B, \overrightarrow{E})$ denote the digraph after the execution of \textsc{Greedy}.
Let $\sigma$ be a linear ordering of $B$ obtained from a topological ordering of $D$.
Let $\operatorname{OPT}(G, \pi)$ denote the optimal objective value of the OSCM instance $(G, \pi)$.
Observe that $\sum_{\{u,v\} \in {B \choose 2}} \min\{c_{u,v}, c_{v,u}\}$ is a lower bound on $\operatorname{OPT}(G,\pi)$.
Therefore, if $c_{u,v} \leq c_{v,u}$ implies $\sigma(u) < \sigma(v)$ for every unordered pair $\{u, v\} \in {B \choose 2}$, then $\sigma$ is an optimal solution.
If $\sigma(u) < \sigma(v)$ and $c_{u,v} > c_{v,u}$ hold for a pair $\{u, v\} \in {B \choose 2}$, then $(u,v) \notin \overrightarrow{E}$ since \textup{Greedy} never adds an arc $(u,v)$ with $c_{u,v} > c_{v,u}$, and $(v,u) \notin \overrightarrow{E}$ since this would contradict $\sigma(u) < \sigma(v)$, 
but there is a directed $u$-$v$-path $v_1v_2 \dots v_k$ in $D$ with $v_1 = u$, $v_k = v$ and $k \geq 3$.
Since $(v_1,v_2), \dots, (v_{k-1},v_k) \in \overrightarrow{E}$, the crossing numbers of the unordered pairs in $\{v_1, v_2, \dots, v_k\}$ fulfill the assumptions of Question~\ref{q:cycle}, and
\[
\frac{c_{v,u}}{c_{u,v}} = \frac{c_{v_k,v_1}}{c_{v_1,v_k}} = \max \left\{ \frac{c_{v_1,v_2}}{c_{v_2,v_1}}, \frac{c_{v_2,v_3}}{c_{v_3,v_2}}, \dots, \frac{c_{v_k,v_1}}{c_{v_1,v_k}} \right\} \geq c,
\]
since the unordered pairs $\{v_i,v_{i+1}\}$ for $i \in [k-1]$ are considered before $\{u,v\} = \{v_1,v_k\}$ in \textsc{Greedy}. 
Thus,
\[
\sum_{\substack{u,v \in B \\ \sigma(u) < \sigma(v)}} c_{u,v} \leq \frac{1}{c} \sum_{\{u,v\} \in {B \choose 2}} \min\{c_{u,v}, c_{v,u}\} \leq \frac{1}{c} \operatorname{OPT}(G, \pi),
\]
meaning that \textsc{Greedy} is indeed a $\frac{1}{c}$-factor approximation algorithm in this case.

Theorem~\ref{thm:cycle} gives lower bounds on the constants $c_t$ in Question~\ref{q:cycle}.

\begin{theorem}[restate=thmcycle]\label{thm:cycle}
Let $(G = (A \,\dot{\cup}\, B, E), \pi)$ be an instance of \textup{OSCM},
and let $v_1, v_2, \dots, v_t \in B$ be $t \geq 3$ distinct vertices that
\begin{enumerate}[(i)]
\item have nonempty neighborhoods, and
\item fulfill the inequality $c_{v_i,v_{i+1}} \leq c_{v_{i+1},v_i}$ for every $i \in [t]$, where $v_{t+1} = v_1$.
%\item fulfill the inequality $c_{v_{i+1},v_i} > 0$ for every $i \in [t]$,
\end{enumerate}
Then
\[
\max \left\{ \frac{c_{v_1,v_2}}{c_{v_2,v_1}}, \frac{c_{v_2,v_3}}{c_{v_3,v_2}}, \dots, \frac{c_{v_t,v_1}}{c_{v_1,v_t}} \right\} \geq \frac{1}{t-1}.
\]
\end{theorem}

If we impose disjoint neighborhoods, we are able to identify the largest possible $c_3$.
\begin{theorem}[restate=thmthreecycle]\label{thm:3cycle}
Let $(G = (A \,\dot{\cup}\, B, E), \pi)$ be an instance of \textup{OSCM},
and let $v_1, v_2, v_3 \in B$ be three distinct vertices that
\begin{enumerate}[(i)]
\item have nonempty, pairwise disjoint neighborhoods, and
\item fulfill the inequalities $c_{v_1,v_2} \leq c_{v_2,v_1}$, $c_{v_2,v_3} \leq c_{v_3,v_2}$, and $c_{v_3,v_1} \leq c_{v_1,v_3}$.
\end{enumerate}
Then
\[
\max \left\{ \frac{c_{v_1,v_2}}{c_{v_2,v_1}}, \frac{c_{v_2,v_3}}{c_{v_3,v_2}}, \frac{c_{v_3,v_1}}{c_{v_1,v_3}} \right\} \geq \phi-1,
\]
where $\phi = \frac{1 + \sqrt{5}}{2}$ is the golden ratio.
This bound is sharp.
\end{theorem}

%With the crossing numbers, OSCM can be seen as a special instance of the \textsc{Linear Ordering Problem (LOP)}:
%Here, we are given an $(k \times k)$-matrix $B = (b_{i,j})$ and we are seeking a linear ordering $\sigma$ of $[k]$ that minimizes $\sum_{\substack{i,j \in [k] \\ \sigma(i) < \sigma(j)}} b_{i,j}$.
%LOP is \APX-hard~\cite{kann1992approximability}, meaning that no polynomial-time approximation scheme (PTAS) exists unless \P{} = \NP.
%On the other hand, neither \APX-hardness nor a PTAS is known for OSCM.

\medskip
The article is structured as follows.
In Section~\ref{sec:pre} we define all required notions.
In Section~\ref{sec:size} we prove Theorem~\ref{thm:size}.
In Section~\ref{sec:unsuited} we prove Theorem~\ref{thm:unsuited} and Corollary~\ref{cor:unsuited}.
In Section~\ref{sec:cycle} we prove Theorem~\ref{thm:cycle}.
At last, in Section~\ref{sec:conclusion}, we give a conclusion.

\section{Preliminaries}\label{sec:pre}
Let $G$ be a graph of order $n$ and size $m$.
A \emph{linear ordering} of a finite set $S = \{s_1, s_2, \dots, s_k\}$ is a bijection $\pi: S \rightarrow [k]$, which we often write as $\pi = s_1s_2 \dots s_k$, where $\pi(i) = s_i$ for all $i \in [k]$.
Given a fixed linear ordering $\pi = v_1v_2 \dots v_n$ of the vertices of $G$ and $k \in [n-1]$, the \emph{linear (edge-)cut} $E_k$ is the set of edges of $G$ having one end in $\{v_1, \dots, v_k\}$ and the other end in $\{v_{k+1}, \dots, v_n\}$.
Note that $E_k$ depends on $\pi$, but since the linear ordering is fixed in the subsequent sections, we do not explicitly state this dependence.
The cutwidth of $G$ with respect to $\pi$ is $\max_{k \in [n-1]} |E_k|$.
The \emph{cutwidth} of $G$ is the smallest integer $w$ with the property that there exists a linear ordering $\pi$ of the vertices of $G$ such that the cutwidth of $G$ with respect to $\pi$ is at most $w$.
In this case, we say that $\pi$ \emph{attains} the cutwidth $w$ of $G$.

Let $(G, \pi)$ be an arbitrary fixed instance of OSCM or OSCM[CW].
For two distinct vertices $u, v \in B$, let the \emph{crossing number} $c_{u,v}$ denote the number of crossings that edges incident to $u$ make with edges incident to $v$ in a linear ordering $\sigma$ of $B$ with $\sigma(u) < \sigma(v)$.
Clearly, the number of crossings only depends on the relative positions between all pairs of distinct vertices of $G$.
Therefore, the number of crossings for a linear ordering $\sigma$ of $B$ is given by
\[
%\operatorname{cr}(G, \pi; \sigma) = 
\sum_{\substack{u,v \in B \\ \sigma(u) < \sigma(v)}} c_{u,v}.
\]

\section{Proof of Theorem~\ref{thm:size}}\label{sec:size}
Throughout Section~\ref{sec:size}, 
let $G$ be a graph of order $n$, size $m$, and cutwidth $w$,
and let $\pi = v_1v_2 \dots v_n$ be a fixed linear ordering of the vertices of $G$ that attains the cutwidth $w$ of $G$.
We prove that the statement for Theorem~\ref{thm:size} holds for $G$.

We begin with two auxiliary lemmas.
For $\ell \in [n-1]$, let $m_\ell$ denote the number of edges $v_iv_j \in E(G)$ with $|i-j|=\ell$.
%We call $\sum_{\ell=1}^{n-1} \ell  m_\ell$ the number of \emph{used slots}.
\begin{lemma}\label{lem:up_mk}
For every $\ell \in [n-1]$, $m_\ell \leq n-\ell$.
\end{lemma}
\begin{proof}
This follows immediately from the linear ordering of the vertices.
\end{proof}
%We call the quantity $(n-1)w$ the number of \emph{available slots}.
%Lemma~\ref{lem:ub_sum_trivial} asserts that terminology is consistent with intuition: The number of used slots is at most the number of available slots.
\begin{lemma}\label{lem:ub_sum_trivial}
$\sum_{\ell=1}^{n-1} \ell m_\ell \leq (n-1)w$.
\end{lemma}
\begin{proof}
Let $v_iv_j$ be an edge of $G$, and let $\ell = |i-j|$.
The edge $v_iv_j$ is in exactly $\ell$ linear cuts $E_k$.
Therefore,
\[
\sum_{\ell=1}^{n-1} \ell m_\ell = \sum_{k=1}^{n-1} |E_k| \leq (n-1)w.
\]
\end{proof}

The upper bound of Lemma~\ref{lem:ub_sum_trivial} can be improved whenever we are considering a graph class for which $w = \Theta(n^2)$.
The idea is that, in this case, the cutwidth $w$ is attained by a linear cut $E_k$ ``somewhere in the middle'' for some $k \approx \frac{n}{2}$, while linear cuts $E_k$ with index $k$ close to $1$ or $n-1$, respectively, are smaller than $w$.
Note that the cutwidth of the complete graph of order $n$ is $\lfloor \frac{n}{2} \rfloor \cdot \lceil \frac{n}{2} \rceil = \Theta(n^2)$.
\begin{lemma}\label{lem:ub_sum}
Let $c \in \left(0, \frac{1}{2}\right]$ be such that $w = c(1-c)n^2$.
If $cn \geq 1$, then
\[
\sum_{\ell=1}^{n-1} \ell m_\ell 
\leq \frac{1}{3}(4c^3 - 6c^2 + 3c)n^3 + \mathcal{O}(n^2).
\]
\end{lemma}
\begin{proof}
Let $c \in \left(0, \frac{1}{2}\right]$ be as in the statement of Lemma~\ref{lem:ub_sum}. 
For $k \in \left[ \lfloor cn \rfloor \right]$, there are at most $k(n-k)$ edges in the linear cut $E_k$.
The same is true for the $\lfloor cn \rfloor$ linear cuts $E_k$ with the $\lfloor cn \rfloor$ largest indices $k$.
The cutwidth $w$ of $G$ is an upper bound for the remaining linear cuts.
Therefore,
\begin{equation}
\begin{aligned}\label{ieq:1}
\sum_{\ell=1}^{n-1} \ell m_\ell 
= \sum_{k=1}^{n-1} |E_k|
&\leq (n - 1 - 2\lfloor cn \rfloor)w + 2 \cdot \sum_{k=1}^{\lfloor cn \rfloor}k(n-k)\\
&\leq (1-2c)nw + w + 2 \cdot \sum_{k=1}^{\lfloor cn \rfloor}k(n-k).
\end{aligned}
\end{equation}
%where the first equality is justified in the proof of Lemma~\ref{lem:ub_sum_trivial}.
%We evaluate the sum using the well-known identities $\sum_{i=1}^{n} i = \frac{n(n+1)}{2}$ and $\sum_{i=1}^{n}i^2 = \frac{n(n+1)(2n+1)}{6}$, which gives
Evaluating and estimating the sum gives
%\marginpar{actually $<$}
\[
\sum_{k=1}^{\lfloor cn \rfloor}k(n-k) 
= \frac{1}{6} \lfloor cn \rfloor (\lfloor cn \rfloor + 1)  (3n - 2\lfloor cn \rfloor - 1)
\leq \frac{1}{6} cn (cn + 1) (3n - 2cn + 1).
\]
Inserting this into (\ref{ieq:1}), substituting $w = c(1-c)n^2$, and simplifying gives
\[
\sum_{\ell=1}^{n-1} \ell m_\ell 
\leq \frac{1}{3} \left(
(4c^3 - 6c^2 + 3c)n^3
+ (-4c^2 + 6c)n^2
+ cn
\right)
\]
which implies the statement of Lemma~\ref{lem:ub_sum}.
\end{proof}

In order to estimate the size $m$ of $G$, we bound the objective function of a linear program.
For this note that $m = \sum_{\ell=1}^{n-1} m_\ell$.
\begin{equation}\tag{$P$}\label{lp:primal1}
\begin{array}{ll@{}ll}
\max \: &\displaystyle\sum_{\ell=1}^{n-1} &m_\ell,&\\
\text{s.t.} \: &&m_\ell \leq n-\ell, &\ell = 1, \dots, n-1,\\
&\displaystyle\sum_{\ell=1}^{n-1} \ell &m_\ell \leq b,&\\
&&m_\ell \geq 0, &\ell = 1, \dots, n-1.
\end{array}
\end{equation}
Here $b$ denotes an arbitrary upper bound of $\sum_{\ell=1}^{n-1}\ell m_\ell$, 
for instance the one of Lemma~\ref{lem:ub_sum_trivial} or Lemma~\ref{lem:ub_sum}.
We introduce the nonnegative dual variables $y_\ell$ for $\ell \in [n-1]$ and $x$ for the constraints of (\ref{lp:primal1}) in the same order as the constraints of (\ref{lp:primal1}).
The weak duality inequality chain of (\ref{lp:primal1}) is
\[
\sum_{\ell=1}^{n-1} m_\ell 
\leq \sum_{\ell=1}^{n-1} m_\ell y_\ell + \sum_{\ell=1}^{n-1}\ell m_\ell x 
= \sum_{\ell=1}^{n-1} (y_\ell + \ell x)m_\ell 
\leq \sum_{\ell=1}^{n-1} (n-\ell)y_\ell + bx.
\]
In particular, the dual of (\ref{lp:primal1}) is the following linear program:
\begin{equation}\tag{$D$}\label{lp:dual1}
\begin{array}{lr@{}ll}
\min \: &\displaystyle\sum_{\ell=1}^{n-1} (n-\ell) y_\ell \,&+\, bx,&\\
\text{s.t.} \: &y_\ell + \ell x \: &\geq 1, &\ell = 1, \dots, n-1,\\
&y_\ell \: &\geq 0, &\ell = 1, \dots, n-1,\\
&x \: &\geq 0.
\end{array}
\end{equation}
By weak duality, the objective value of any feasible solution of (\ref{lp:dual1}) is an upper bound to the objective value of (\ref{lp:primal1}), and therefore an upper bound on $m$.

We are in a position to prove Theorem~\ref{thm:size}.
\thmubs*
\begin{proof}
Let $d = \left\lceil \frac{\sqrt{w}}{a} \right\rceil$, where $a > 0$ is a fixed constant be determined such that $d \leq n$.
A feasible solution of (\ref{lp:dual1}) is given by
\[
x = \frac{a}{\sqrt{w}} 
\quad\text{and}\quad 
y_\ell = \begin{cases}
	1 - \frac{\ell a}{\sqrt{w}},&\text{if $\ell < d$},\\
	0,&\text{otherwise.}
\end{cases}
\]
We estimate the objective value of (\ref{lp:dual1}) from above.
By the discussion before Theorem~\ref{thm:size}, this suffices to prove the stated bounds on $m$.
For the sum we have
\begin{align*}
\sum_{\ell=1}^{n-1} (n-\ell) y_\ell
&= \sum_{\ell=1}^{d - 1} (n - \ell) \left(1 - \frac{\ell a}{\sqrt{w}}\right)
\\
&= \frac{(d - 1) \left(2 a d^{2} - 3 a d n - a d - 3 d \sqrt{w} + 6 n \sqrt{w}\right)}{6 \sqrt{w}}
\\
&\leq \frac{\frac{\sqrt{w}}{a} \left(2 a \left(\frac{\sqrt{w}}{a} + 1\right)^{2} - 3 a \frac{\sqrt{w}}{a} n - a \frac{\sqrt{w}}{a} - 3 \frac{\sqrt{w}}{a} \sqrt{w} + 6 n \sqrt{w}\right)}{6 \sqrt{w}}
\\
&= \frac{n\sqrt{w}}{2a} - \frac{w}{6a^2} + \frac{\sqrt{w}}{2a} + \frac{1}{3}.
\end{align*}
%using $\sum_{i=1}^{n} i = \frac{n(n+1)}{2}$ and $\sum_{i=1}^{n}i^2 = \frac{n(n+1)(2n+1)}{6}$ again.

For (i), we use $a = \frac{1}{\sqrt{2}}$, and $b = (n-1)w$ from Lemma~\ref{lem:ub_sum_trivial} in the linear programs (\ref{lp:primal1}) and (\ref{lp:dual1}).
Note that, since the complete graph maximizes $w$ for a given order $n$ and
\[
a > \frac{1}{2} 
\geq \frac{\sqrt{\left\lfloor n/2 \right\rfloor \left\lceil n/2 \right\rceil}}{n}
\geq \frac{\sqrt{w}}{n},
\]
we have $\frac{\sqrt{w}}{a} \leq n$, and since $n$ is an integer, also $d \leq n$.
Therefore the choice of $a$ is valid.
Now,
\[
bx 
= (n-1)w \cdot \frac{a}{\sqrt{w}} 
\leq an\sqrt{w}.
\]
Note that, for the choice of $a$ and $w$ sufficiently large, $\frac{w}{6a^2} = w/3 \geq \sqrt{w/2} = \frac{\sqrt{w}}{2a}$.
We conclude that
\[
\sum_{\ell=1}^{n-1} (n-\ell)y_\ell + bx \leq \left( \frac{1}{2a} + a \right)n\sqrt{w} + \mathcal{O}(1).
\]
Observe that the term $\frac{1}{2a} + a$ attains its minimum $\sqrt{2}$ at $a = \frac{1}{\sqrt{2}}$.
This completes the proof for (i).

For (ii), assume that $c \in \left(0, \frac{1}{2}\right]$ is such that $w = c(1-c)n^2$ and $cn \geq 1$. 
Let $C = \frac{1}{3}(4c^3 - 6c^2 + 3c)$ and $C'=\sqrt{c(1-c)}$.
Here we use $b = Cn^3 + \mathcal{O}(n^2)$ from Lemma~\ref{lem:ub_sum}.
Then we have
\[
bx 
\leq \left( Cn^3 + \mathcal{O}(n^2) \right) \cdot \frac{a}{\sqrt{w}} 
= \frac{Ca}{C'}n^2 + \mathcal{O}(n).
\]
We conclude that in this case
\[
\sum_{\ell=1}^{n-1} (n-\ell)y_\ell + bx \leq \left(\frac{C'}{2a} + \frac{Ca}{C'} - \frac{C'^2}{6a^2}\right)n^2 + \mathcal{O}(n).
\]
Choosing $a = \frac{1}{2}$, which implies $d \leq n$ again, gives the desired result.
This completes the proof for (ii).
\end{proof}

\begin{figure}[ht]
\centering
\begin{tikzpicture}
\begin{scope}[shift={(-4.25,0)}]
\begin{axis}[
	axis equal image,
	axis lines=left,
	xlabel= \(c\),
	ylabel= \({a(c)}\),
	xmin=0, xmax=0.52,
	ymin=0.5, ymax=0.72,
	xtick={0, 0.1, 0.2, 0.3, 0.4, 0.5},
	ytick={0.5, 0.6, 0.7},
	xmajorgrids,
	ymajorgrids
]
\addplot[
	color=black,
	%thick
] coordinates {
(0.001,0.696679649797151)(0.002,0.692423937628369)(0.003,0.689181993591304)(0.004,0.686464210745808)(0.005,0.684081024301118)(0.006,0.681935253255515)(0.007,0.679969190583974)(0.008,0.678145251764742)(0.009000000000000001,0.676437340709063)(0.01,0.674826464410294)(0.011,0.673298290244764)(0.012,0.671841687496674)(0.013000000000000001,0.670447807907391)(0.014,0.669109480104734)(0.015,0.667820796120481)(0.016,0.666576820399069)(0.017,0.665373379671243)(0.018000000000000002,0.664206907816993)(0.019,0.663074329092226)(0.02,0.661972968726055)(0.021,0.660900483434832)(0.022,0.659854806685289)(0.023,0.658834105052596)(0.024,0.657836743043110)(0.025,0.656861254458109)(0.026000000000000002,0.655906318871029)(0.027,0.654970742144860)(0.028,0.654053440172814)(0.029,0.653153425213712)(0.03,0.652269794333369)(0.031,0.651401719568495)(0.032,0.650548439509472)(0.033,0.649709252059672)(0.034,0.648883508176420)(0.035,0.648070606435730)(0.036000000000000004,0.647269988292111)(0.037,0.646481133927829)(0.038,0.645703558604473)(0.039,0.644936809444502)(0.04,0.644180462582434)(0.041,0.643434120635100)(0.042,0.642697410448367)(0.043000000000000003,0.641969981084286)(0.044,0.641251502018049)(0.045,0.640541661518647)(0.046,0.639840165190871)(0.047,0.639146734659428)(0.048,0.638461106378632)(0.049,0.637783030553312)(0.05,0.637112270158515)(0.051000000000000004,0.636448600047170)(0.052000000000000005,0.635791806136257)(0.053,0.635141684663218)(0.054,0.634498041505323)(0.055,0.633860691555623)(0.056,0.633229458149835)(0.057,0.632604172539170)(0.058,0.631984673404692)(0.059000000000000004,0.631370806409281)(0.06,0.630762423783692)(0.061,0.630159383943606)(0.062,0.629561551134880)(0.063,0.628968795104502)(0.064,0.628380990795007)(0.065,0.627798018060356)(0.066,0.627219761401450)(0.067,0.626646109719653)(0.068,0.626076956086855)(0.069,0.625512197530729)(0.07,0.624951734833982)(0.07100000000000001,0.624395472346493)(0.07200000000000001,0.623843317809355)(0.073,0.623295182189902)(0.074,0.622750979526895)(0.075,0.622210626785122)(0.076,0.621674043718718)(0.077,0.621141152742565)(0.078,0.620611878811216)(0.079,0.620086149304799)(0.08,0.619563893921410)(0.081,0.619045044575573)(0.082,0.618529535302327)(0.083,0.618017302166588)(0.084,0.617508283177414)(0.085,0.617002418206878)(0.08600000000000001,0.616499648913221)(0.08700000000000001,0.615999918668043)(0.088,0.615503172487250)(0.089,0.615009356965539)(0.09,0.614518420214196)(0.091,0.614030311802011)(0.092,0.613544982699107)(0.093,0.613062385223529)(0.094,0.612582472990412)(0.095,0.612105200863590)(0.096,0.611630524909500)(0.097,0.611158402353242)(0.098,0.610688791536689)(0.099,0.610221651878516)(0.1,0.609756943836056)(0.101,0.609294628868864)(0.10200000000000001,0.608834669403921)(0.10300000000000001,0.608377028802361)(0.10400000000000001,0.607921671327664)(0.105,0.607468562115229)(0.106,0.607017667143241)(0.107,0.606568953204794)(0.108,0.606122387881177)(0.109,0.605677939516278)(0.11,0.605235577192052)(0.111,0.604795270704987)(0.112,0.604356990543537)(0.113,0.603920707866449)(0.114,0.603486394481973)(0.115,0.603054022827880)(0.116,0.602623565952276)(0.117,0.602194997495150)(0.11800000000000001,0.601768291670648)(0.11900000000000001,0.601343423250014)(0.12,0.600920367545183)(0.121,0.600499100392989)(0.122,0.600079598139969)(0.123,0.599661837627721)(0.124,0.599245796178804)(0.125,0.598831451583151)(0.126,0.598418782084976)(0.127,0.598007766370144)(0.128,0.597598383553999)(0.129,0.597190613169614)(0.13,0.596784435156457)(0.131,0.596379829849452)(0.132,0.595976777968416)(0.133,0.595575260607855)(0.134,0.595175259227117)(0.135,0.594776755640861)(0.136,0.594379732009868)(0.137,0.593984170832132)(0.138,0.593590054934269)(0.139,0.593197367463192)(0.14,0.592806091878067)(0.14100000000000001,0.592416211942520)(0.14200000000000002,0.592027711717105)(0.14300000000000002,0.591640575552004)(0.14400000000000002,0.591254788079963)(0.145,0.590870334209447)(0.146,0.590487199118016)(0.147,0.590105368245897)(0.148,0.589724827289761)(0.149,0.589345562196690)(0.15,0.588967559158330)(0.151,0.588590804605212)(0.152,0.588215285201257)(0.153,0.587840987838435)(0.154,0.587467899631586)(0.155,0.587096007913397)(0.156,0.586725300229523)(0.157,0.586355764333848)(0.158,0.585987388183893)(0.159,0.585620159936343)(0.16,0.585254067942717)(0.161,0.584889100745145)(0.162,0.584525247072282)(0.163,0.584162495835323)(0.164,0.583800836124138)(0.165,0.583440257203509)(0.166,0.583080748509478)(0.167,0.582722299645790)(0.168,0.582364900380432)(0.169,0.582008540642276)(0.17,0.581653210517803)(0.171,0.581298900247917)(0.17200000000000001,0.580945600224850)(0.17300000000000001,0.580593300989147)(0.17400000000000002,0.580241993226725)(0.17500000000000002,0.579891667766021)(0.176,0.579542315575202)(0.177,0.579193927759462)(0.178,0.578846495558379)(0.179,0.578500010343344)(0.18,0.578154463615053)(0.181,0.577809847001074)(0.182,0.577466152253462)(0.183,0.577123371246446)(0.184,0.576781495974163)(0.185,0.576440518548467)(0.186,0.576100431196773)(0.187,0.575761226259967)(0.188,0.575422896190366)(0.189,0.575085433549730)(0.19,0.574748831007314)(0.191,0.574413081337982)(0.192,0.574078177420356)(0.193,0.573744112235016)(0.194,0.573410878862740)(0.195,0.573078470482794)(0.196,0.572746880371250)(0.197,0.572416101899361)(0.198,0.572086128531960)(0.199,0.571756953825907)(0.2,0.571428571428571)(0.201,0.571100975076345)(0.202,0.570774158593197)(0.203,0.570448115889261)(0.20400000000000001,0.570122840959453)(0.20500000000000002,0.569798327882124)(0.20600000000000002,0.569474570817744)(0.20700000000000002,0.569151564007620)(0.20800000000000002,0.568829301772633)(0.209,0.568507778512020)(0.21,0.568186988702172)(0.211,0.567866926895465)(0.212,0.567547587719118)(0.213,0.567228965874076)(0.214,0.566911056133921)(0.215,0.566593853343805)(0.216,0.566277352419410)(0.217,0.565961548345931)(0.218,0.565646436177083)(0.219,0.565332011034130)(0.22,0.565018268104938)(0.221,0.564705202643046)(0.222,0.564392809966764)(0.223,0.564081085458285)(0.224,0.563770024562824)(0.225,0.563459622787771)(0.226,0.563149875701869)(0.227,0.562840778934405)(0.228,0.562532328174426)(0.229,0.562224519169967)(0.23,0.561917347727302)(0.231,0.561610809710206)(0.232,0.561304901039242)(0.233,0.560999617691061)(0.234,0.560694955697714)(0.23500000000000001,0.560390911145987)(0.23600000000000002,0.560087480176748)(0.23700000000000002,0.559784658984308)(0.23800000000000002,0.559482443815803)(0.23900000000000002,0.559180830970584)(0.24,0.558879816799624)(0.241,0.558579397704939)(0.242,0.558279570139027)(0.243,0.557980330604310)(0.244,0.557681675652605)(0.245,0.557383601884591)(0.246,0.557086105949304)(0.247,0.556789184543635)(0.248,0.556492834411841)(0.249,0.556197052345079)(0.25,0.555901835180934)(0.251,0.555607179802975)(0.252,0.555313083140316)(0.253,0.555019542167187)(0.254,0.554726553902515)(0.255,0.554434115409529)(0.256,0.554142223795357)(0.257,0.553850876210646)(0.258,0.553560069849191)(0.259,0.553269801947571)(0.26,0.552980069784801)(0.261,0.552690870681986)(0.262,0.552402202001991)(0.263,0.552114061149124)(0.264,0.551826445568819)(0.265,0.551539352747336)(0.266,0.551252780211472)(0.267,0.550966725528272)(0.268,0.550681186304763)(0.269,0.550396160187684)(0.27,0.550111644863232)(0.271,0.549827638056821)(0.272,0.549544137532842)(0.273,0.549261141094430)(0.274,0.548978646583255)(0.275,0.548696651879306)(0.276,0.548415154900690)(0.277,0.548134153603441)(0.278,0.547853645981332)(0.279,0.547573630065704)(0.28,0.547294103925298)(0.281,0.547015065666093)(0.28200000000000003,0.546736513431159)(0.28300000000000003,0.546458445400515)(0.28400000000000003,0.546180859790993)(0.28500000000000003,0.545903754856117)(0.28600000000000003,0.545627128885981)(0.28700000000000003,0.545350980207143)(0.28800000000000003,0.545075307182529)(0.289,0.544800108211330)(0.29,0.544525381728928)(0.291,0.544251126206820)(0.292,0.543977340152542)(0.293,0.543704022109618)(0.294,0.543431170657509)(0.295,0.543158784411563)(0.296,0.542886862022989)(0.297,0.542615402178828)(0.298,0.542344403601933)(0.299,0.542073865050962)(0.3,0.541803785320379)(0.301,0.541534163240456)(0.302,0.541264997677294)(0.303,0.540996287532843)(0.304,0.540728031744944)(0.305,0.540460229287359)(0.306,0.540192879169830)(0.307,0.539925980438134)(0.308,0.539659532174156)(0.309,0.539393533495957)(0.31,0.539127983557871)(0.311,0.538862881550594)(0.312,0.538598226701289)(0.313,0.538334018273702)(0.314,0.538070255568279)(0.315,0.537806937922312)(0.316,0.537544064710065)(0.317,0.537281635342938)(0.318,0.537019649269622)(0.319,0.536758105976278)(0.32,0.536497004986714)(0.321,0.536236345862579)(0.322,0.535976128203564)(0.323,0.535716351647625)(0.324,0.535457015871195)(0.325,0.535198120589427)(0.326,0.534939665556442)(0.327,0.534681650565583)(0.328,0.534424075449687)(0.329,0.534166940081369)(0.33,0.533910244373312)(0.331,0.533653988278576)(0.332,0.533398171790917)(0.333,0.533142794945116)(0.334,0.532887857817325)(0.335,0.532633360525426)(0.336,0.532379303229400)(0.337,0.532125686131717)(0.338,0.531872509477732)(0.339,0.531619773556099)(0.34,0.531367478699201)(0.341,0.531115625283597)(0.342,0.530864213730479)(0.343,0.530613244506148)(0.34400000000000003,0.530362718122509)(0.34500000000000003,0.530112635137575)(0.34600000000000003,0.529862996156001)(0.34700000000000003,0.529613801829623)(0.34800000000000003,0.529365052858015)(0.34900000000000003,0.529116749989084)(0.35000000000000003,0.528868894019655)(0.35100000000000003,0.528621485796101)(0.352,0.528374526214972)(0.353,0.528128016223666)(0.354,0.527881956821104)(0.355,0.527636349058434)(0.356,0.527391194039759)(0.357,0.527146492922880)(0.358,0.526902246920078)(0.359,0.526658457298905)(0.36,0.526415125383003)(0.361,0.526172252552962)(0.362,0.525929840247191)(0.363,0.525687889962817)(0.364,0.525446403256622)(0.365,0.525205381746001)(0.366,0.524964827109956)(0.367,0.524724741090117)(0.368,0.524485125491797)(0.369,0.524245982185083)(0.37,0.524007313105951)(0.371,0.523769120257444)(0.372,0.523531405710843)(0.373,0.523294171606918)(0.374,0.523057420157201)(0.375,0.522821153645292)(0.376,0.522585374428225)(0.377,0.522350084937860)(0.378,0.522115287682342)(0.379,0.521880985247582)(0.38,0.521647180298810)(0.381,0.521413875582157)(0.382,0.521181073926315)(0.383,0.520948778244226)(0.384,0.520716991534841)(0.385,0.520485716884945)(0.386,0.520254957471018)(0.387,0.520024716561195)(0.388,0.519794997517261)(0.389,0.519565803796723)(0.39,0.519337138954970)(0.391,0.519109006647486)(0.392,0.518881410632152)(0.393,0.518654354771621)(0.394,0.518427843035788)(0.395,0.518201879504349)(0.396,0.517976468369420)(0.397,0.517751613938300)(0.398,0.517527320636301)(0.399,0.517303593009689)(0.4,0.517080435728734)(0.401,0.516857853590879)(0.402,0.516635851524019)(0.403,0.516414434589920)(0.404,0.516193607987734)(0.405,0.515973377057698)(0.406,0.515753747284941)(0.40700000000000003,0.515534724303443)(0.40800000000000003,0.515316313900165)(0.40900000000000003,0.515098522019340)(0.41000000000000003,0.514881354766922)(0.41100000000000003,0.514664818415242)(0.41200000000000003,0.514448919407826)(0.41300000000000003,0.514233664364432)(0.41400000000000003,0.514019060086302)(0.41500000000000004,0.513805113561609)(0.41600000000000004,0.513591831971170)(0.417,0.513379222694383)(0.418,0.513167293315431)(0.419,0.512956051629764)(0.42,0.512745505650864)(0.421,0.512535663617315)(0.422,0.512326534000203)(0.423,0.512118125510837)(0.424,0.511910447108872)(0.425,0.511703508010783)(0.426,0.511497317698724)(0.427,0.511291885929894)(0.428,0.511087222746282)(0.429,0.510883338484926)(0.43,0.510680243788719)(0.431,0.510477949617715)(0.432,0.510276467261055)(0.433,0.510075808349520)(0.434,0.509875984868709)(0.435,0.509677009172997)(0.436,0.509478894000204)(0.437,0.509281652487091)(0.438,0.509085298185762)(0.439,0.508889845080940)(0.44,0.508695307608317)(0.441,0.508501700673967)(0.442,0.508309039674880)(0.443,0.508117340520829)(0.444,0.507926619657526)(0.445,0.507736894091292)(0.446,0.507548181415273)(0.447,0.507360499837383)(0.448,0.507173868210139)(0.449,0.506988306062474)(0.45,0.506803833633782)(0.451,0.506620471910326)(0.452,0.506438242664347)(0.453,0.506257168495921)(0.454,0.506077272878036)(0.455,0.505898580205110)(0.456,0.505721115845256)(0.457,0.505544906196750)(0.458,0.505369978749112)(0.459,0.505196362149300)(0.46,0.505024086273484)(0.461,0.504853182305244)(0.462,0.504683682820714)(0.463,0.504515621881524)(0.464,0.504349035136718)(0.465,0.504183959934313)(0.466,0.504020435444276)(0.467,0.503858502793904)(0.468,0.503698205217766)(0.46900000000000003,0.503539588223750)(0.47000000000000003,0.503382699778109)(0.47100000000000003,0.503227590511713)(0.47200000000000003,0.503074313951392)(0.47300000000000003,0.502922926779779)(0.47400000000000003,0.502773489129003)(0.47500000000000003,0.502626064913423)(0.47600000000000003,0.502480722208864)(0.47700000000000004,0.502337533686757)(0.47800000000000004,0.502196577113881)(0.47900000000000004,0.502057935930998)(0.48,0.501921699927387)(0.481,0.501787966032087)(0.482,0.501656839249781)(0.483,0.501528433777160)(0.484,0.501402874346190)(0.485,0.501280297858327)(0.486,0.501160855395787)(0.487,0.501044714728772)(0.488,0.500932063489444)(0.489,0.500823113257949)(0.49,0.500718104928866)(0.491,0.500617315925209)(0.492,0.500521070170249)(0.493,0.500429752339304)(0.494,0.500343829097811)(0.495,0.500263882473854)(0.496,0.500190666088178)(0.497,0.500125209600449)(0.498,0.500069043299375)(0.499,0.500024821611436)(0.5,0.500000000000000)
};
\end{axis}
\end{scope}
\begin{scope}[shift={(4.25,0)}]
\begin{axis}[
	axis equal image,
	axis lines=left,
	xlabel= \(c\),
	ylabel= \({f(c), g(c)}\),
	xmin=0, xmax=0.52,
	ymin=0, ymax=0.52,
	xtick={0, 0.1, 0.2, 0.3, 0.4, 0.5},
	ytick={0, 0.1, 0.2, 0.3, 0.4, 0.5},
	xmajorgrids,
	ymajorgrids
]
\addplot[
	domain=0.000001:0.5,
	samples=500,
	color=black,
	%thick
]
{(x*(4*x^2 + 4*x*sqrt(x*(1-x)) - 12*x - 4*sqrt(x*(1-x)) + 9)/(6*sqrt(x*(1-x)))};
\addplot[
	color=black,
	dashed
] coordinates {
(0.001,0.0443388707897048)(0.002,0.0624405077772612)(0.003,0.0762063900559429)(0.004,0.0877206822512060)(0.005,0.0977916538465298)(0.006,0.106834320115452)(0.007,0.115095503560343)(0.008,0.122736206972130)(0.009000000000000001,0.129868391024015)(0.01,0.136573666606450)(0.011,0.142913713625088)(0.012,0.148936505911206)(0.013000000000000001,0.154680238101712)(0.014,0.160175913662349)(0.015,0.165449113138330)(0.016,0.170521239384274)(0.017,0.175410417336271)(0.018000000000000002,0.180132158748530)(0.019,0.184699862872592)(0.02,0.189125200030480)(0.021,0.193418409930359)(0.022,0.197588536813887)(0.023,0.201643617061672)(0.024,0.205590830509115)(0.025,0.209436623705737)(0.026000000000000002,0.213186811229940)(0.027,0.216846659656722)(0.028,0.220420957678744)(0.029,0.223914075075417)(0.03,0.227330012625844)(0.031,0.230672444611004)(0.032,0.233944755208420)(0.033,0.237150069819891)(0.034,0.240291282169544)(0.035,0.243371077850597)(0.036000000000000004,0.246391954874227)(0.037,0.249356241674752)(0.038,0.252266112946125)(0.039,0.255123603621095)(0.04,0.257930621252824)(0.041,0.260688957016884)(0.042,0.263400295517228)(0.043000000000000003,0.266066223551504)(0.044,0.268688237967803)(0.045,0.271267752725485)(0.046,0.273806105256649)(0.047,0.276304562211177)(0.048,0.278764324656963)(0.049,0.281186532797248)(0.05,0.283572270258868)(0.051000000000000004,0.285922567998240)(0.052000000000000005,0.288238407866003)(0.053,0.290520725866136)(0.054,0.292770415141024)(0.055,0.294988328710171)(0.056,0.297175281987000)(0.057,0.299332055095381)(0.058,0.301459395005061)(0.059000000000000004,0.303558017503029)(0.06,0.305628609016011)(0.061,0.307671828297623)(0.062,0.309688307992303)(0.063,0.311678656086854)(0.064,0.313643457259359)(0.065,0.315583274134196)(0.066,0.317498648451049)(0.067,0.319390102155030)(0.068,0.321258138414328)(0.069,0.323103242571215)(0.07,0.324925883031661)(0.07100000000000001,0.326726512098368)(0.07200000000000001,0.328505566751549)(0.073,0.330263469381430)(0.074,0.332000628476072)(0.075,0.333717439267816)(0.076,0.335414284341357)(0.077,0.337091534206198)(0.078,0.338749547836021)(0.079,0.340388673177277)(0.08,0.342009247629121)(0.081,0.343611598496658)(0.082,0.345196043419283)(0.083,0.346762890775782)(0.084,0.348312440067716)(0.085,0.349844982282502)(0.08600000000000001,0.351360800237494)(0.08700000000000001,0.352860168906262)(0.088,0.354343355728193)(0.089,0.355810620902451)(0.09,0.357262217667239)(0.091,0.358698392565268)(0.092,0.360119385696260)(0.093,0.361525430957237)(0.094,0.362916756271333)(0.095,0.364293583805786)(0.096,0.365656130179722)(0.097,0.367004606662328)(0.098,0.368339219361941)(0.099,0.369660169406567)(0.1,0.370967653116290)(0.101,0.372261862168032)(0.10200000000000001,0.373542983753053)(0.10300000000000001,0.374811200727605)(0.10400000000000001,0.376066691757082)(0.105,0.377309631454008)(0.106,0.378540190510206)(0.107,0.379758535823411)(0.108,0.380964830618641)(0.109,0.382159234564575)(0.11,0.383341903885190)(0.111,0.384512991466896)(0.112,0.385672646961382)(0.113,0.386821016884393)(0.114,0.387958244710620)(0.115,0.389084470964905)(0.116,0.390199833309920)(0.117,0.391304466630496)(0.11800000000000001,0.392398503114756)(0.11900000000000001,0.393482072332194)(0.12,0.394555301308845)(0.121,0.395618314599682)(0.122,0.396671234358353)(0.123,0.397714180404388)(0.124,0.398747270287977)(0.125,0.399770619352445)(0.126,0.400784340794495)(0.127,0.401788545722348)(0.128,0.402783343211838)(0.129,0.403768840360581)(0.13,0.404745142340270)(0.131,0.405712352447199)(0.132,0.406670572151065)(0.133,0.407619901142148)(0.134,0.408560437376907)(0.135,0.409492277122075)(0.136,0.410415514997307)(0.137,0.411330244016439)(0.138,0.412236555627411)(0.139,0.413134539750918)(0.14,0.414024284817816)(0.14100000000000001,0.414905877805360)(0.14200000000000002,0.415779404272293)(0.14300000000000002,0.416644948392850)(0.14400000000000002,0.417502592989706)(0.145,0.418352419565909)(0.146,0.419194508335847)(0.147,0.420028938255264)(0.148,0.420855787050390)(0.149,0.421675131246183)(0.15,0.422487046193745)(0.151,0.423291606096926)(0.152,0.424088884038144)(0.153,0.424878952003468)(0.154,0.425661880906958)(0.155,0.426437740614321)(0.156,0.427206599965889)(0.157,0.427968526798938)(0.158,0.428723587969392)(0.159,0.429471849372911)(0.16,0.430213375965395)(0.161,0.430948231782925)(0.162,0.431676479961153)(0.163,0.432398182754167)(0.164,0.433113401552841)(0.165,0.433822196902698)(0.166,0.434524628521283)(0.167,0.435220755315091)(0.168,0.435910635396029)(0.169,0.436594326097463)(0.17,0.437271883989832)(0.171,0.437943364895866)(0.17200000000000001,0.438608823905407)(0.17300000000000001,0.439268315389851)(0.17400000000000002,0.439921893016223)(0.17500000000000002,0.440569609760894)(0.176,0.441211517922952)(0.177,0.441847669137238)(0.178,0.442478114387063)(0.179,0.443102904016597)(0.18,0.443722087742969)(0.181,0.444335714668060)(0.182,0.444943833290010)(0.183,0.445546491514452)(0.184,0.446143736665470)(0.185,0.446735615496297)(0.186,0.447322174199759)(0.187,0.447903458418469)(0.188,0.448479513254786)(0.189,0.449050383280535)(0.19,0.449616112546503)(0.191,0.450176744591718)(0.192,0.450732322452511)(0.193,0.451282888671368)(0.194,0.451828485305587)(0.195,0.452369153935733)(0.196,0.452904935673904)(0.197,0.453435871171814)(0.198,0.453962000628689)(0.199,0.454483363798997)(0.2,0.455000000000000)(0.201,0.455511948119142)(0.202,0.456019246621275)(0.203,0.456521933555733)(0.20400000000000001,0.457020046563243)(0.20500000000000002,0.457513622882697)(0.20600000000000002,0.458002699357772)(0.20700000000000002,0.458487312443414)(0.20800000000000002,0.458967498212179)(0.209,0.459443292360445)(0.21,0.459914730214494)(0.211,0.460381846736461)(0.212,0.460844676530163)(0.213,0.461303253846810)(0.214,0.461757612590595)(0.215,0.462207786324167)(0.216,0.462653808274000)(0.217,0.463095711335646)(0.218,0.463533528078884)(0.219,0.463967290752768)(0.22,0.464397031290568)(0.221,0.464822781314622)(0.222,0.465244572141080)(0.223,0.465662434784566)(0.224,0.466076399962740)(0.225,0.466486498100774)(0.226,0.466892759335742)(0.227,0.467295213520924)(0.228,0.467693890230024)(0.229,0.468088818761314)(0.23,0.468480028141690)(0.231,0.468867547130660)(0.232,0.469251404224249)(0.233,0.469631627658834)(0.234,0.470008245414907)(0.23500000000000001,0.470381285220767)(0.23600000000000002,0.470750774556143)(0.23700000000000002,0.471116740655751)(0.23800000000000002,0.471479210512785)(0.23900000000000002,0.471838210882341)(0.24,0.472193768284788)(0.241,0.472545909009061)(0.242,0.472894659115914)(0.243,0.473240044441097)(0.244,0.473582090598490)(0.245,0.473920822983168)(0.246,0.474256266774424)(0.247,0.474588446938726)(0.248,0.474917388232635)(0.249,0.475243115205660)(0.25,0.475565652203070)(0.251,0.475885023368652)(0.252,0.476201252647430)(0.253,0.476514363788324)(0.254,0.476824380346776)(0.255,0.477131325687321)(0.256,0.477435222986123)(0.257,0.477736095233456)(0.258,0.478033965236161)(0.259,0.478328855620041)(0.26,0.478620788832229)(0.261,0.478909787143517)(0.262,0.479195872650637)(0.263,0.479479067278512)(0.264,0.479759392782468)(0.265,0.480036870750406)(0.266,0.480311522604945)(0.267,0.480583369605526)(0.268,0.480852432850480)(0.269,0.481118733279065)(0.27,0.481382291673474)(0.271,0.481643128660801)(0.272,0.481901264714987)(0.273,0.482156720158725)(0.274,0.482409515165341)(0.275,0.482659669760644)(0.276,0.482907203824744)(0.277,0.483152137093846)(0.278,0.483394489162016)(0.279,0.483634279482909)(0.28,0.483871527371491)(0.281,0.484106252005709)(0.28200000000000003,0.484338472428158)(0.28300000000000003,0.484568207547710)(0.28400000000000003,0.484795476141118)(0.28500000000000003,0.485020296854604)(0.28600000000000003,0.485242688205414)(0.28700000000000003,0.485462668583352)(0.28800000000000003,0.485680256252298)(0.289,0.485895469351691)(0.29,0.486108325898000)(0.291,0.486318843786170)(0.292,0.486527040791043)(0.293,0.486732934568766)(0.294,0.486936542658168)(0.295,0.487137882482128)(0.296,0.487336971348913)(0.297,0.487533826453505)(0.298,0.487728464878904)(0.299,0.487920903597412)(0.3,0.488111159471899)(0.301,0.488299249257056)(0.302,0.488485189600624)(0.303,0.488668997044604)(0.304,0.488850688026461)(0.305,0.489030278880297)(0.306,0.489207785838017)(0.307,0.489383225030476)(0.308,0.489556612488609)(0.309,0.489727964144548)(0.31,0.489897295832721)(0.311,0.490064623290935)(0.312,0.490229962161452)(0.313,0.490393327992035)(0.314,0.490554736236999)(0.315,0.490714202258231)(0.316,0.490871741326207)(0.317,0.491027368620990)(0.318,0.491181099233216)(0.319,0.491332948165069)(0.32,0.491482930331238)(0.321,0.491631060559868)(0.322,0.491777353593493)(0.323,0.491921824089959)(0.324,0.492064486623335)(0.325,0.492205355684814)(0.326,0.492344445683593)(0.327,0.492481770947760)(0.328,0.492617345725148)(0.329,0.492751184184195)(0.33,0.492883300414784)(0.331,0.493013708429076)(0.332,0.493142422162328)(0.333,0.493269455473711)(0.334,0.493394822147104)(0.335,0.493518535891886)(0.336,0.493640610343723)(0.337,0.493761059065330)(0.338,0.493879895547239)(0.339,0.493997133208551)(0.34,0.494112785397677)(0.341,0.494226865393073)(0.342,0.494339386403965)(0.343,0.494450361571066)(0.34400000000000003,0.494559803967285)(0.34500000000000003,0.494667726598424)(0.34600000000000003,0.494774142403868)(0.34700000000000003,0.494879064257273)(0.34800000000000003,0.494982504967234)(0.34900000000000003,0.495084477277957)(0.35000000000000003,0.495184993869916)(0.35100000000000003,0.495284067360503)(0.352,0.495381710304673)(0.353,0.495477935195581)(0.354,0.495572754465207)(0.355,0.495666180484984)(0.356,0.495758225566406)(0.357,0.495848901961638)(0.358,0.495938221864119)(0.359,0.496026197409152)(0.36,0.496112840674495)(0.361,0.496198163680938)(0.362,0.496282178392878)(0.363,0.496364896718889)(0.364,0.496446330512283)(0.365,0.496526491571664)(0.366,0.496605391641476)(0.367,0.496683042412553)(0.368,0.496759455522647)(0.369,0.496834642556968)(0.37,0.496908615048707)(0.371,0.496981384479556)(0.372,0.497052962280223)(0.373,0.497123359830942)(0.374,0.497192588461980)(0.375,0.497260659454131)(0.376,0.497327584039213)(0.377,0.497393373400558)(0.378,0.497458038673492)(0.379,0.497521590945817)(0.38,0.497584041258283)(0.381,0.497645400605062)(0.382,0.497705679934205)(0.383,0.497764890148109)(0.384,0.497823042103967)(0.385,0.497880146614225)(0.386,0.497936214447022)(0.387,0.497991256326638)(0.388,0.498045282933929)(0.389,0.498098304906762)(0.39,0.498150332840443)(0.391,0.498201377288146)(0.392,0.498251448761332)(0.393,0.498300557730169)(0.394,0.498348714623943)(0.395,0.498395929831473)(0.396,0.498442213701512)(0.397,0.498487576543153)(0.398,0.498532028626231)(0.399,0.498575580181711)(0.4,0.498618241402088)(0.401,0.498660022441773)(0.402,0.498700933417477)(0.403,0.498740984408595)(0.404,0.498780185457586)(0.405,0.498818546570346)(0.406,0.498856077716584)(0.40700000000000003,0.498892788830188)(0.40800000000000003,0.498928689809595)(0.40900000000000003,0.498963790518154)(0.41000000000000003,0.498998100784485)(0.41100000000000003,0.499031630402839)(0.41200000000000003,0.499064389133453)(0.41300000000000003,0.499096386702902)(0.41400000000000003,0.499127632804447)(0.41500000000000004,0.499158137098388)(0.41600000000000004,0.499187909212400)(0.417,0.499216958741885)(0.418,0.499245295250304)(0.419,0.499272928269517)(0.42,0.499299867300118)(0.421,0.499326121811769)(0.422,0.499351701243527)(0.423,0.499376615004177)(0.424,0.499400872472552)(0.425,0.499424482997863)(0.426,0.499447455900017)(0.427,0.499469800469939)(0.428,0.499491525969889)(0.429,0.499512641633778)(0.43,0.499533156667484)(0.431,0.499553080249162)(0.432,0.499572421529557)(0.433,0.499591189632314)(0.434,0.499609393654283)(0.435,0.499627042665824)(0.436,0.499644145711117)(0.437,0.499660711808461)(0.438,0.499676749950577)(0.439,0.499692269104905)(0.44,0.499707278213911)(0.441,0.499721786195376)(0.442,0.499735801942700)(0.443,0.499749334325192)(0.444,0.499762392188370)(0.445,0.499774984354249)(0.446,0.499787119621640)(0.447,0.499798806766436)(0.448,0.499810054541907)(0.449,0.499820871678991)(0.45,0.499831266886581)(0.451,0.499841248851816)(0.452,0.499850826240370)(0.453,0.499860007696740)(0.454,0.499868801844533)(0.455,0.499877217286758)(0.456,0.499885262606110)(0.457,0.499892946365256)(0.458,0.499900277107129)(0.459,0.499907263355211)(0.46,0.499913913613821)(0.461,0.499920236368406)(0.462,0.499926240085830)(0.463,0.499931933214659)(0.464,0.499937324185457)(0.465,0.499942421411070)(0.466,0.499947233286926)(0.467,0.499951768191321)(0.468,0.499956034485717)(0.46900000000000003,0.499960040515035)(0.47000000000000003,0.499963794607954)(0.47100000000000003,0.499967305077210)(0.47200000000000003,0.499970580219895)(0.47300000000000003,0.499973628317761)(0.47400000000000003,0.499976457637525)(0.47500000000000003,0.499979076431177)(0.47600000000000003,0.499981492936291)(0.47700000000000004,0.499983715376341)(0.47800000000000004,0.499985751961014)(0.47900000000000004,0.499987610886535)(0.48,0.499989300335992)(0.481,0.499990828479667)(0.482,0.499992203475371)(0.483,0.499993433468783)(0.484,0.499994526593804)(0.485,0.499995490972905)(0.486,0.499996334717494)(0.487,0.499997065928287)(0.488,0.499997692695688)(0.489,0.499998223100182)(0.49,0.499998665212742)(0.491,0.499999027095249)(0.492,0.499999316800925)(0.493,0.499999542374797)(0.494,0.499999711854172)(0.495,0.499999833269146)(0.496,0.499999914643155)(0.497,0.499999963993558)(0.498,0.499999989332295)(0.499,0.499999998666621)(0.5,0.500000000000000)
};
\end{axis}
\end{scope}
\end{tikzpicture}
\caption{Left: The choice of $a$ depending on $c$ for the minimum of $\frac{C'}{2a} + \frac{Ca}{C'} + \frac{C'^2}{6a^2}$.
Right: $f(c) = \frac{c(4c^2 + 4c\sqrt{c(1-c)} - 12c - 4\sqrt{c(1-c)} + 9)}{6 \sqrt{c(1-c)}}$ (straight) from the bound of Theorem~\ref{thm:size}~(ii) compared to the improved bound $g(c)$ (dashed) on $\left[0, \frac{1}{2}\right]$; see Remark~\ref{rem:size}.}\label{fig:min}
\end{figure}
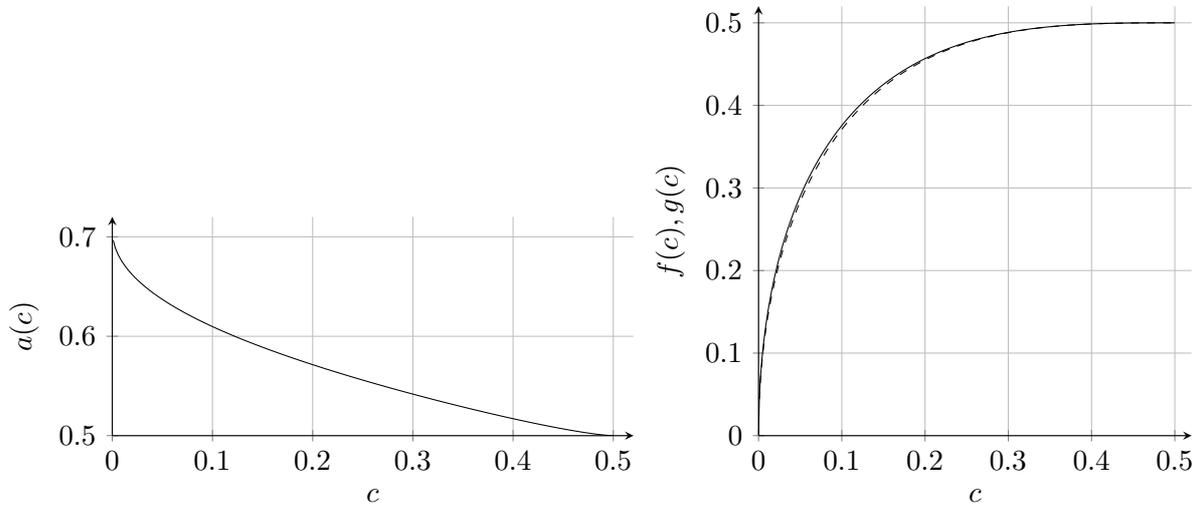

\begin{remark}\label{rem:size}
An attentive reader might have noticed that in the proof of Theorem~\ref{thm:size}~(ii) the choice $a = \frac{1}{2}$ is not the minimum of $\frac{C'}{2a} + \frac{Ca}{C'} + \frac{C'^2}{6a^2}$ for every $c \in \left(0,\frac{1}{2}\right]$, but it is very ``close'' to it.
By applying the first-order condition to this term, we find that choosing $a$ in dependence of $c$ as shown on the left side of Figure~\ref{fig:min} gives the minimum.
For the sake of simplicity, we do not state the function $a(c)$ explicitly.
The resulting improved bound $g(c)$ (dashed) is compared to the bound $f(c)$ stated in Theorem~\ref{thm:size}~(ii) (straight) on the right side of Figure~\ref{fig:min}.
Again, for the sake of exposition, we do not state the function $g(c)$ explicitly.
As claimed, the difference is ``marginal''.
\end{remark}

We close this section with an example that shows that the bound of Theorem~\ref{thm:size}~(i) is asymptotically sharp.
Let $k$ be a positive integer, $n$ an integer larger than $k$, and $w = 1+2+\dots+k = \frac{k(k+1)}{2}$.
Let $G$ be the graph with vertices $v_1, v_2, \dots, v_n$ and edges $v_iv_j$ for $i,j \in [n]$ with $|i-j| \leq k$; see Figure~\ref{fig:cw-extremal}.
The linear ordering $\pi = v_1v_2\dots v_n$ witnesses that $G$ has cutwidth at most $w$ while the size of $G$ is 
\[
n-1 + n-2 + \dots + n-k = kn - \frac{k(k+1)}{2} = \left(\sqrt{w + \frac{1}{4}} - \frac{1}{2}\right)n - w.
\]
We believe that this is the extremal configuration.

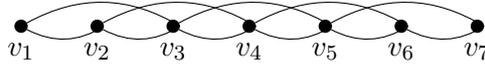
\begin{figure}[h]
\centering
\begin{tikzpicture}
\foreach \i in {1,...,7}{
	\node[dot,label=below:{$v_{\i}$}] (v\i) at (\i,0) {};
};
\foreach \i in {1,...,6}{
	\pgfmathparse{\i+1}
	\draw[bend right] (v\i) to (v\pgfmathresult);
};
\foreach \i in {1,...,5}{
	\pgfmathparse{\i+2}
	\draw[bend left] (v\i) to (v\pgfmathresult);
};
\end{tikzpicture}
\caption{The constructed graph for $k=2$ and $n=7$.}\label{fig:cw-extremal}
\end{figure}

\section{Proof of Theorem~\ref{thm:unsuited} and Corollary~\ref{cor:unsuited}}\label{sec:unsuited}
Let $(G, \pi)$ be an arbitrary fixed instance of OSCM[CW].
In this section we provide a sharp upper bound on the number of unsuited pairs in terms $n$ and $w$, where $w$ is the cutwidth of $G$ with respect to $\pi$.
%We also provide an example to show that this is essentially best possible.
We start right away with the main theorem of this section.
\thmunsuited*
\begin{proof}
If $\{u, v\} \subseteq B$ is an unsuited pair of two distinct vertices, that is, $c_{u,v} > 0$ and $c_{v,u} > 0$,
then, for any linear ordering $\sigma$ of $B$, there is at least one edge incident to $u$ and one edge incident to $v$ that cross in the linear ordering $\sigma$.
In particular, this holds for $\sigma = \pi|_B$.

Let $\pi = v_1v_2\dots v_n$, $u \in B$, and $f(u)$ denote the number of $v \in B \setminus \{u\}$ with $\pi(u) < \pi(v)$ such that $\{u, v\}$ is an unsuited pair.
Clearly, the number of unsuited pairs is equal to
$\sum_{u \in B} f(u)$.
Let $r$ denote the neighbor of $u$ in $G$ that maximizes $\pi(r)$.
Let $k = \pi(u)$ if $\pi(u) < \pi(r)$ and $k = \pi(u)-1$ otherwise.
If $v \in B \setminus \{u\}$ with $\pi(u) < \pi(v)$ is such that $\{u, v\}$ is unsuited, then $\pi(r) > 1$ and there is one edge incident to $v$ that crosses the edge $ur$ and is in the linear cut $E_k$.
Since $ur$ is in $E_k$ and $\pi$ has cutwidth $w$ with respect to $G$, there are at most $w-1$ such vertices $v$, that is, $f(u) \leq w-1$.

The above upper bound can be improved in some cases as follows.
Let 
\[
\operatorname{pos}(\pi; u) := |\{v \in B: \pi(v) \leq \pi(u)\}|,
\]
which is a bijection between $B$ and $[|B|]$.
If $\operatorname{pos}(\pi; u) = i \geq |B| - (w-1)$ and $j = |B| - i$, then there are $j$ vertices $v \in B \setminus \{u\}$ with $\pi(u) < \pi(v)$.
Therefore, $f(u) \leq j \leq w-1$ is true in this case.

Altogether, the number of unsuited pairs is bounded from above by
\begin{align*}
\sum_{i=1}^{|B|} f(\operatorname{pos}^{-1}(\pi; i)) 
&\leq \sum_{i=1}^{|B|-w} (w-1) + \sum_{i=|B|-(w-1)}^{|B|} (|B|-i)\\
&= (|B|-w)(w-1) + {w \choose 2}\\
&= |B|(w-1) - {w \choose 2}.
\end{align*}

For the sharpness,
consider the bipartite graph $G$ with partite sets $A$ and $B$ in Figure~\ref{fig:bip}.
Let $\pi$ denote the linear ordering of the vertices of $G$ that sorts the vertices according to their $x$-coordinates as in Figure~\ref{fig:bip}.
That means, the leftmost vertex $v$ in Figure~\ref{fig:bip} has $\pi(v)=1$, and so on.
Clearly, $\pi$ has cutwidth $6$ with respect to $G$.
It is easy to see that $(G,\pi)$ matches the bound of Theorem~\ref{thm:unsuited}; that is, there are
\[
|B|(6-1) - {6 \choose 2}
\]
unsuited pairs in $B$.
With the same type of construction, one finds graphs $G$ of arbitrarily large order and linear orderings of arbitrary cutwidth with respect to $G$ that match the bound of Theorem~\ref{thm:unsuited}.

\begin{figure}[h]
\centering
\begin{tikzpicture}[scale=0.95]
	%A
	\foreach \a in {1, ..., 17} {
		\node[dot] (A\a) at (\a-9,1) {};
	};
	%B
	\foreach \b in {1, ..., 12} {
		\node[dot] (B\b) at (\b-6.5,0) {};
	};
	%edges leaning left
	\foreach \i in {1, ..., 12} {
		\draw (A\i) -- (B\i);
	};
	%edges leaning right
	\foreach \i in {1, ..., 11} {
		\pgfmathparse{\i+6}
		\draw (A\pgfmathresult) -- (B\i);
	};
	%rectangles
	\draw[rounded corners, dashed, darkgray] (-8.25,1.25) rectangle (8.25,0.75);
	\node (A) at (-7.5,1.5) {$A$};
	\draw[rounded corners, dashed, darkgray] (-5.75,0.25) rectangle (5.75,-0.25);
	\node (B) at (-5,-0.5) {$B$};
	%extension
	%\node[dot, label=above:$a$] (A18) at (9,1) {};
	%\node[dot, label=below:$b$] (B13) at (6.5,0) {};
	%\draw[dashed] (A13) -- (B13);
	%\draw[dashed] (A18) -- (B13);
\end{tikzpicture}
\caption{A bipartite graph $G$ with partite sets $A$ and $B$.
If $\pi$ is the linear ordering of $V(G)$ induced by the $x$-coordinates of the vertices, then the bound of Theorem~\ref{thm:unsuited} is tight for $G$ and $\pi$.
%The vertices $a$ and $b$ and its incident edges are \emph{not} part of $G$; they indicate its extension.
}
\label{fig:bip}
\end{figure}
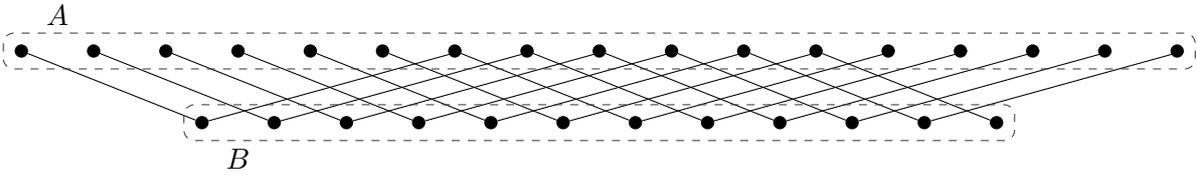
\end{proof}

We close this section by stating Corollary~\ref{cor:unsuited}, which is an immediate consequence of Theorem~\ref{thm:unsuited}.
We remark that it can be determined in polynomial time if a pair is suited or not~\cite{dujmovic2004efficient}.

\corunsuited*

\section{Proof of Theorem~\ref{thm:cycle} and Theorem~\ref{thm:3cycle}}\label{sec:cycle}
We directly begin with the proof of Theorem~\ref{thm:cycle}.
\thmcycle*
\begin{proof}
We assume, for a contradiction, that $\max \left\{ \frac{c_{v_1,v_2}}{c_{v_2,v_1}}, \frac{c_{v_2,v_3}}{c_{v_3,v_2}}, \dots, \frac{c_{v_t,v_1}}{c_{v_1,v_t}} \right\} < \frac{1}{t-1}$.
In this proof, for ease of notation, $u_i$ always denotes an arbitrary neighbor of $v_i$ for $i \in [t]$.
If we sum or quantify over ``$u_i$'', then this is meant as ``$u_i \in N_G(v_i)$''.

Let $i,j \in [t]$ with $i \neq j$ be arbitrary.
Recall that the edges $u_iv_i$ and $u_jv_j$ cross in a linear ordering $\sigma$ of $B$ with $\sigma(v_i) < \sigma(v_j)$ if and only if $\pi(u_i) > \pi(u_j)$.
We define
\[
a_{u_i,u_j} = \begin{cases}
	1, &\text{if $u_iv_i$ and $u_jv_j$ cross in a linear ordering $\sigma$ of $B$ with $\sigma(v_i) < \sigma(v_j)$,}\\
	0, &\text{otherwise}.
\end{cases}
\]
From the definition it follows that, for all $i,j \in [t]$ with $i \neq j$, $c_{v_i, v_j} = \sum_{u_i,u_j} a_{u_i,u_j}$, and $a_{u_i,u_j} = a_{u_j,u_i} = 0$ if (and only if) $u_i=u_j$.
Let $d_i$ denote the degree of $v_i$ in $G$ for $i \in [t]$, $D_{i,j} = \prod_{\substack{k \in [t]\\k \notin \{i,j\}}} d_k$ for $i,j \in [t]$, and $S$ denote the number of sets $\{u_1,u_2, \dots, u_t\}$ of cardinality at least 2.
%For ease of notation, let $t+1=1$ in the indices for the remaining proof.
It also follows from the definition that 
\begin{itemize}
\item $a_{u_1,u_2} + a_{u_2,u_3} + \dots + a_{u_t,u_1} \geq 1$ for all $u_1,u_2,\dots,u_t$ with $|\{u_1,u_2,\dots,u_t\}| \geq 2$, and
$a_{u_1,u_2} + a_{u_2,u_3} + \dots + a_{u_t,u_1} = 0$ for all $u_1,u_2,\dots,u_t$ with $u_1=u_2=\dots=u_t$,
which implies $D_{1,2}c_{v_1,v_2} + D_{2,3}c_{v_2,v_3} + \dots + D_{1,t}c_{v_t,v_1} \geq S$ after summing all these inequalities up, and
\item $a_{u_2,u_1} + a_{u_3,u_2} + \dots + a_{u_1,u_t} \leq t-1$ for all $u_1,u_2,\dots,u_t$ with $|\{u_1,u_2,\dots,u_t\}| \geq 2$, and
$a_{u_2,u_1} + a_{u_3,u_2} + \dots + a_{u_1,u_t} = 0$ for all $u_1,u_2,\dots,u_t$ with $u_1=u_2=\dots=u_t$,
which implies $D_{1,2}c_{v_2,v_1} + D_{2,3}c_{v_3,v_2} + \dots + D_{1,t}c_{v_1,v_t} \leq (t-1)S$ after summing these inequalities up.
\end{itemize}

At this point, since $c_{v_{i+1},v_i} > c_{v_i,v_{i+1}}(t-1)$ holds for every $i \in [t]$ by assumption, where $v_{t+1}=v_1$, we arrive at the contradiction
\begin{align*}
(t-1)S
&\geq D_{1,2}c_{v_2,v_1} + D_{2,3}c_{v_3,v_2} + \dots + D_{1,t}c_{v_1,v_t}
\\
&> (D_{1,2}c_{v_1,v_2} + D_{2,3}c_{v_2,v_3} + \dots + D_{1,t}c_{v_t,v_1})(t-1)
\\
&\geq (t-1)S,
\end{align*}
which completes the proof.
\end{proof}

We continue with the proof of Theorem~\ref{thm:3cycle}.

\thmthreecycle*
\begin{proof}
\begin{figure}
\centering
\begin{tikzpicture}
\node[dot,label=below:$v_3$] (v3) at (-3,0) {};
\node[dot,label=below:$v_2$] (v2) at (0,0) {};
\node[dot,label=below:$v_1$] (v1) at (3,0) {};
\foreach \i in {1, ..., 5} {
	\node[dot,label=above:{$u_{\i}$}] (u\i) at (2*\i-6,2) {};
	\draw (v2) to (u\i);
};
\foreach \i in {0, ..., 5} {
	\node[inner sep=0,label=above:{$A_{\i}$}] (a\i) at (2*\i-5.333,2) {};
	\draw[thick] (v3) to (a\i);
};
\foreach \i in {0, ..., 5} {
	\node[inner sep=0,label=above:{$B_{\i}$}] (b\i) at (-2*\i+5.333,2) {};
	\draw[thick] (v1) to (b\i);
};
\end{tikzpicture}
\caption{Definition of the variables $A_i$ and $B_j$ for $d=5$.}\label{fig:thm-3cycle}
\end{figure}
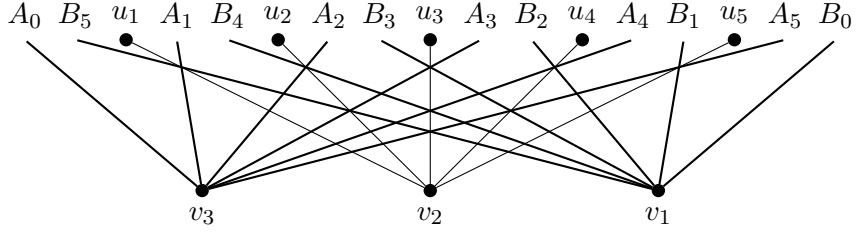
We assume that $\max \left\{ \frac{c_{v_1,v_2}}{c_{v_2,v_1}}, \frac{c_{v_2,v_3}}{c_{v_3,v_2}} \right\} < \phi - 1$ and prove $\frac{c_{v_3,v_1}}{c_{v_1,v_3}} \geq \phi - 1$.
Let $d_1$, $d=d_2$ and $d_3$ denote the degree of $v_1$, $v_2$ and $v_3$ in $G$, respectively.
Since we assume that the neighborhoods are pairwise disjoint, we have $c_{v_i,v_j} + c_{v_j,v_i} = d_id_j$ for $i,j \in [3]$ with $i \neq j$.
With this, the assumption is equivalent to
\[
\frac{c_{v_1,v_2}}{c_{v_2,v_1}} < \phi - 1
\quad\Longleftrightarrow\quad
d_1d - c_{v_2,v_1} = c_{v_1,v_2} < (\phi - 1)c_{v_2,v_1}
\quad\Longleftrightarrow\quad
\frac{c_{v_2,v_1}}{d_1} > \frac{d}{\phi}
\]
and, similarly, $\frac{c_{v_3,v_2}}{d_3} > \frac{d}{\phi}$.  
The goal is to prove
\[
\frac{c_{v_3,v_1}}{c_{v_1,v_3}} \geq \phi - 1
\quad\Longleftrightarrow\quad
c_{v_3,v_1} \geq (\phi - 1)c_{v_1,v_3} = (\phi - 1)(d_1d_3 - c_{v_3,v_1})
\quad\Longleftrightarrow\quad
\frac{c_{v_3,v_1}}{d_1d_3} \geq \frac{\phi-1}{\phi}.
\]
Let $N_G(v_2) = \{u_1, u_2, \dots, u_d\}$ denote the neighborhood of $v_2$ in $G$, where $\pi(u_1) < \pi(u_2) < \dots < \pi(u_d)$.
Moreover, for $i \in \{0, 1, \dots, d\}$, let
\begin{itemize}
\item $A_i = |\{w \in N_G(v_3): \pi(u_i) < \pi(w) < \pi(u_{i+1})\}|$, and
\item $B_i = |\{w \in N_G(v_1): \pi(u_{d-i}) < \pi(w) < \pi(u_{d-i+1})\}|$,
\end{itemize}
where we set $\pi(u_0) = -\infty$ and $\pi(u_{d+1}) = +\infty$ for ease of notation; see Figure~\ref{fig:thm-3cycle} for an illustration.
Note that $c_{v_2,v_1} = \sum_{i=1}^{d}iB_i$, $c_{v_3,v_2} = \sum_{i=1}^{d}iA_i$, and $c_{v_3,v_1} \geq \sum_{i=1}^{d} \sum_{j=1}^{i} A_iB_{d+1-j}$.
Consider the optimization problem~(\ref{op}).

%\begin{figure}[h]
\begin{equation}\tag{$O$}\label{op}
\begin{array}{lr@{}ll}
\min \: &\displaystyle\sum_{i=1}^{d} \displaystyle\sum_{j=1}^{i} a_i&b_{d+1-j},&\\
\text{s.t.} \: &\displaystyle\sum_{i=1}^{d} ia_i \: &\geq \frac{d}{\phi},&\\
&\displaystyle\sum_{i=1}^{d} ib_i \: &\geq \frac{d}{\phi},&\\
&\displaystyle\sum_{i=0}^{d} a_i \: &= 1,&\\
&\displaystyle\sum_{i=0}^{d} b_i \: &= 1,&\\
&a_i&\geq 0, &i=0,1,\dots,d,\\
&b_i&\geq 0, &i=0,1,\dots,d.
\end{array}
\end{equation}
%\end{figure}

With the normalizations $a_i = \frac{A_i}{d_3}$ and $b_i = \frac{B_i}{d_1}$ for $i \in \{0,1,\dots,d\}$, it suffices to prove that $\operatorname{OPT}(\text{\ref{op}}) \geq \frac{\phi-1}{\phi}$.
Since the set of feasible solutions of (\ref{op}) is compact and the objective function is continuous, the optimization problem (\ref{op}) has an optimal solution.
Note that, in any optimal solution of (\ref{op}), we may assume that the constraints $\sum_{i=1}^d ia_i \geq \frac{d}{\phi}$ and $\sum_{i=1}^d ib_i \geq \frac{d}{\phi}$ are fulfilled with equality:
Otherwise, if say $\sum_{i=1}^d ia_i > \frac{d}{\phi}$, there exist $j \in [d]$ and $\epsilon > 0$ such that $a_j' = a_j - \epsilon$, $a_0' = a_0 + \epsilon$, and $a_\ell' = a_\ell$ for $\ell \in [d] \setminus \{j\}$ together the $b_i$ would constitute a feasible solution of (\ref{op}) with an equal or smaller objective value and $\sum_{i=1}^d ia_i' = \frac{d}{\phi}$.
Therefore, the following discussion only considers feasible solutions that fulfill these inequality constraints of (\ref{op}) with equality.

Suppose that $a_i$ and $b_i$ for $i \in \{0,1,\dots,d\}$ denote an optimal solution of (\ref{op}) with as few positive variables as possible.
\begin{claim}\label{claim:pos}
At most two $a_i$ and at most two $b_i$ are positive.
\end{claim}
\begin{proof}
For $d < 2$, the statement is clear.
Let $d \geq 2$ and assume, for a contradiction, that $a_i$, $a_j$ and $a_k$ are positive for $i,j,k \in \{0,1,\dots,d\}$ with $i<j<k$.
For $\epsilon \in \mathbb{R} \setminus \{0\}$ of small enough absolute value,
let $\delta = \frac{j-i}{k-j}\epsilon$,
and define $a_i' =  a_i + \epsilon$, $a_j' = a_j-\epsilon-\delta$, $a_k' = a_k + \delta$, and $a_\ell' = a_\ell$ for $\ell \in \{0,1,\dots,d\} \setminus \{i,j,k\}$.
Clearly, $a_i'$ and $b_i$ form a feasible solution of (\ref{op}).
If $\Delta$ denotes the change of the objective function under this shift, then
\[
\Delta 
= \epsilon\left(\sum_{\ell = j+1}^{k} b_{d+1-\ell}\right) - \delta\left(\sum_{\ell=i+1}^{j} b_{d+1-\ell}\right)
= \epsilon \left(1 - \frac{j-i}{k-j}\right) \left( \sum_{\ell = j+1}^{k} b_{d+1-\ell} - \sum_{\ell=i+1}^{j} b_{d+1-\ell} \right).
\]
This shows that it is possible to choose $\epsilon$ such that $\Delta \leq 0$, and the $a_i'$ and $b_i$ form a feasible solution of (\ref{op}) with a smaller number of positive variables, a contradiction to the choice of the $a_i$ and $b_i$.
This proves that at most two $a_i$ are positive.
The proof that at most two $b_i$ are positive is analogous.
\end{proof}

Supplementing Claim~\ref{claim:pos}, note that at least two $a_i$ and at least two $b_i$ must be positive since $\frac{d}{\phi}$ is irrational.
Let $p,q,s,t \in \{0,1,\dots,d\}$ with $p < q$ and $s < t$ be such that $a_p$, $a_q$, $b_s$ and $b_t$ are positive. 
The remaining variables are zero.
At this point, since we assumed that all constraints of (\ref{op}) are fulfilled with equality, the positive variables are uniquely determined and given by $a_p = f(p,q)$, $a_q = g(p,q)$, $b_s = f(s,t)$ and $b_t = g(s,t)$, where
\[
f(x,y) = \frac{1}{y-x}\left(y - \frac{d}{\phi}\right) 
\quad\text{and}\quad
g(x,y) = \frac{1}{y-x}\left(\frac{d}{\phi} - x\right).
\]
Since all variables are nonnegative, we have $p,s \leq \frac{d}{\phi} \leq q,t$.
With these inequalities and the partial derivatives
\begin{align*}
\frac{\partial f}{\partial x} = \frac{\phi y - d}{\phi (y-x)^2},\qquad
\frac{\partial g}{\partial x} = \frac{d - \phi y}{\phi (y-x)^2},\\
\frac{\partial f}{\partial y} = \frac{d - \phi x}{\phi (y-x)^2},\qquad
\frac{\partial g}{\partial y} = \frac{\phi x - d}{\phi (y-x)^2},
\end{align*}
we see that $a_p$ and $b_s$ are monotone increasing in the indices,
while $a_q$ and $b_t$ are monotone decreasing in the indices.
Depending on the concrete values of these indices, the objective function of (\ref{op}) takes different forms.

\medskip\noindent
\textbf{Case 1:} $p+s > d$. In this case, the objective function of (\ref{op}) is simply
\[
a_p(b_s + b_t) + a_q(b_s + b_t) = a_p + a_q = 1 = \operatorname{OPT}(\text{\ref{op}}) > \frac{\phi-1}{\phi}
\]
due to the constraints.

\medskip\noindent
\textbf{Case 2:} $p+s \leq d$. We distinguish four subcases.

\medskip\noindent
\textbf{Case 2a:} $s+q > d$ and $p+t > d$. In this subcase, the objective function of (\ref{op}) is $a_pb_t + a_q = a_qb_s + b_t$.
If $q < d$, the variables $a_p' = a_p + \epsilon$, $a_q' = a_q - \epsilon - \delta$, $a_d' = \delta$, and $a_\ell' = a_\ell$ for $\ell \in \{0,1,\dots,d-1\} \setminus \{p,q\}$ with $\epsilon > 0$ small enough and $\delta = \frac{q-p}{d-q}\epsilon$ would, together with the $b_i$, constitute a feasible solution of (\ref{op}) with a smaller objective value, since $a_pb_s$ is not a summand of the objective function.
This implies $q = d$ and, similarly, $t = d$.
Substituting this into the objective function gives
\[
F(p,s) = \frac{(\phi-1) d \left(d - s \phi\right) + \phi \left(d - s\right) \left(d - p \phi\right)}{\phi^{2} \left(d - p\right) \left(d - s\right)}.
\]
With
\[
\frac{\partial F}{\partial p} = \frac{-(\phi-1)^{2} d^{2}}{\phi^{2} \left(d - p\right)^{2} \left(d - s\right)}
\quad\text{and}\quad
\frac{\partial F}{\partial s} = \frac{-(\phi-1)^{2} d^{2}}{\phi^{2} \left(d - p\right) \left(d - s\right)^{2}}
\]
we see that $F$ is monotone decreasing in both variables.
Since $\frac{d}{2} < \frac{d}{\phi}$, the inequality $p+s \leq d$ is fulfilled with equality.
With $p,s \leq \frac{d}{\phi}$ this implies $d - \frac{d}{\phi} \leq p \leq \frac{d}{\phi}$.
Now,
\[
\frac{\partial F(p, d-p)}{\partial p} = \frac{(\phi-1)^{2} d^{2} \left(d - 2 p\right)}{p^{2} \phi^{2} \left(d - p\right)^{2}}
\]
is positive for $p \in \left[d - \frac{d}{\phi}, \frac{d}{2}\right)$, zero for $p = \frac{d}{2}$, and negative for $p \in \left(\frac{d}{2},\frac{d}{\phi}\right]$.
This implies that a minimum of $F(p,d-p)$ as a function of $p$ on $\left[d - \frac{d}{\phi}, \frac{d}{\phi}\right]$ is attained at the border.
We get
\[
\operatorname{OPT}(\text{\ref{op}})
\geq F\left(d - \frac{d}{\phi}, \frac{d}{\phi}\right) 
= F\left(\frac{d}{\phi}, d - \frac{d}{\phi}\right) 
= 2 - \phi 
= \frac{\phi-1}{\phi}.
\]

\medskip\noindent
\textbf{Case 2b:} $s+q > d$ and $p+t \leq d$. In this subcase, the objective function of (\ref{op}) is $a_q$.
Note that we have $p \leq d - t \leq d - \frac{d}{\phi}$.
Since $a_q$ is monotone decreasing in $p$ and $q$, choosing $p = d - \frac{d}{\phi}$ and $q = d$ implies $\operatorname{OPT}(\text{\ref{op}}) \geq 2-\phi = \frac{\phi-1}{\phi}$.

\medskip\noindent
\textbf{Case 2c:} $s+q \leq d$ and $p+t > d$. This subcase is analogous to Case~2b.

\medskip\noindent
\textbf{Case 2d:} $s+q \leq d$ and $p+t \leq d$. In this subcase, the objective function of (\ref{op}) is $a_qb_t$.
Since $a_q$ and $b_t$ are monotone decreasing in all indices, we have that equality holds in $s \leq d-q$ and $p \leq d-t$.
Let $F(q,t)$ denote the function after substituting $s=d-q$ and $p=d-t$ into the objective function.
The gradient of $F(q,t)$ is nonzero on $\left[\frac{d}{\phi},d\right]^2$.
This means that $F(q,t)$ attains its minimum on $\left[\frac{d}{\phi},d\right]^2$ at the border.
By symmetry, we may assume that $q \in \left\{\frac{d}{\phi},d\right\}$.
If $q = \frac{d}{\phi}$, then
\[
F\left(\frac{d}{\phi},t\right) = \frac{d \left(2 - \phi\right)}{\phi t - (\phi - 1) d}
\]
is monotone decreasing as a function of $t$, since $\phi t \geq d > (\phi-1)d$.
Choosing $t = d$ gives
\[
\operatorname{OPT}(\text{\ref{op}}) \geq F\left(\frac{d}{\phi},d\right) = 2-\phi = \frac{\phi-1}{\phi}.
\]
If $q = d$, then
\[
\frac{\partial F(d,t)}{\partial t} = \frac{d \left(2 (\phi-1) d - \phi t\right)}{t^{3} \phi^2},
\]
which is positive for $t \in \left[\frac{d}{\phi},\frac{2(\phi-1)d}{\phi}\right)$, zero for $t = \frac{2(\phi-1)d}{\phi}$, and negative for $t \in \left(\frac{2(\phi-1)d}{\phi},d\right]$.
This implies that a minimum of $F(d, t)$ as a function of $t$ on $[\frac{d}{\phi},d]$ is attained at the border.
We get 
\[
\operatorname{OPT}(\text{\ref{op}}) \geq F\left(d,\frac{d}{\phi}\right) = F(d,d) = 2-\phi = \frac{\phi-1}{\phi}.
\]

\medskip
This establishes $\frac{c_{v_3,v_1}}{c_{v_1,v_3}} \geq \phi - 1$.
The sharpness of this bound readily follows from configurations constructed from the stated minima of (\ref{op}) in the subcases of Case~2.
This completes the proof of Theorem~\ref{thm:3cycle}.
\end{proof}

\section{Conclusion}\label{sec:conclusion}
In Theorem~\ref{thm:size} we gave two upper bounds on the size of a graph in terms of its order and cutwidth.
Although we know that the bound of Theorem~\ref{thm:size} is asymptotically sharp, we believe that it is possible to improve it.

In Theorem~\ref{thm:unsuited} we gave a sharp upper bound on the number of unsuited pairs of an instance $(G = (A \dot{\cup} B, E), \pi)$ of OSCM[CW] in terms of $|B|$ and the cutwidth of $G$ with respect to $\pi$.
In Corollary~\ref{cor:unsuited} we concluded from this that OSCM[CW] is solvable in $\mathcal{O}^*(2^{|B|w})$-time by simply enumerating all relative positions of unordered pairs in $B$.
Of course, such an algorithm considers also illegal relative positions of unordered pairs in $B$ that do not correspond to a linear ordering of $B$.
Therefore, an $\mathcal{O}^*(\alpha^{|B|w})$-time algorithm for OSCM[CW] with $\alpha \in (1,2)$ as small as possible would be nice to have.
Recall that there is no \FPT-algorithm for OSCM[CW] unless $\P = \NP$ as explained in Section~\ref{sec:intro}.

In Theorem~\ref{thm:cycle} and Theorem~\ref{thm:3cycle} we partially answer Question~\ref{q:cycle} and, to this effect, Question~\ref{q:greedy}.
We saved the best for the end, as Question~\ref{q:cycle} is the most interesting open problem in our opinion.
To conclude the article, we briefly sketch a connection between the number of substrings of a string and Question~\ref{q:cycle}.
Let $B$ be a finite set of $t$ characters,
let $s = s_1s_2 \dots s_a$ be a string consisting of the $t$ different characters in $B$, 
and let $E$ denote the set of edges $\{i, s_i\}$ for $i \in [a]$.
Then $(([a] \cup B, E), \operatorname{id})$ is an instance of OSCM.
Observe that, for two distinct character $u,v \in B$, the number of substrings ``$vu$'' of $s$ is equal to the crossing number $c_{u,v}$.
Conversely, if $(G = (A \dot{\cup} B, E), \operatorname{\pi})$ is an instance of OSCM, where the neighborhoods of the vertices in $B$ are pairwise disjoint, then we can construct a corresponding string $s = s_1s_2 \dots s_{|A|}$.
The substrings of a string formulation really represents the core problem of Question~\ref{q:cycle}.
We are convinced that, with this formulation, a good flag algebraist has a decent chance to fully answer Question~\ref{q:cycle}.

%\bibliographystyle{plain}
%\bibliography{references}

\end{document}